\documentclass[notitlepage,a4paper,superscriptaddress,10pt,tightenlines,twocolumn,prl,nofootinbib]{revtex4-2}
\usepackage[utf8]{inputenc}
\usepackage{graphicx}
\usepackage{physics, amsmath, amssymb, amsthm, units, dsfont, bm}
\usepackage{mathtools, amsfonts, mathrsfs, bbm}
\usepackage{tikz}
\usepackage{xcolor}

\definecolor{dark-gray}{rgb}{.35,.55,.55}
\definecolor{dark-blue}{rgb}{.0,.0,.6}
\usepackage[colorlinks=true,linkcolor=dark-blue,citecolor=dark-blue,urlcolor=dark-blue]{hyperref}

\newcommand{\mean}[1]{\langle #1 \rangle}

\newcommand{\GHZ}{GHZ}
\newcommand{\diag}{{\mathrm{diag}}}
\newcommand{\id}{\mathds{1}}
\newcommand{\esf}{\hat{F}}
\renewcommand{\tr}{\mathrm{tr}}

\newtheorem{theorem}{Theorem}
\newtheorem{corollary}[theorem]{Corollary}

\newtheorem{lemma}[theorem]{Lemma}

\makeatletter
\def\maketitle{
\@author@finish
\title@column\titleblock@produce
\suppressfloats[t]}
\makeatother

\begin{document}

\title{Certifying the Topology of Quantum Networks: Theory and Experiment}

\author{Lisa T. Weinbrenner}
\affiliation{Naturwissenschaftlich-Technische Fakult\"at, Universit\"at Siegen, Walter-Flex-Stra\ss{}e 3, 57068 Siegen, Germany}

\author{Nidhin Prasannan}
\affiliation{Paderborn University, Integrated Quantum Optics,\looseness=-1{ }Institute for Photonic Quantum Systems (PhoQS), Warburger Stra\ss{}e 100, 33098 Paderborn, Germany}

\author{Kiara Hansenne}
\affiliation{Naturwissenschaftlich-Technische Fakult\"at, Universit\"at Siegen, Walter-Flex-Stra\ss{}e 3, 57068 Siegen, Germany}

\author{Sophia Denker}
\affiliation{Naturwissenschaftlich-Technische Fakult\"at, Universit\"at Siegen, Walter-Flex-Stra\ss{}e 3, 57068 Siegen, Germany}

\author{Jan Sperling}
\affiliation{Paderborn University, Theoretical Quantum Science,\looseness=-1{ }Institute for Photonic Quantum Systems (PhoQS), Warburger Stra\ss{}e 100, 33098 Paderborn, Germany}

\author{Benjamin Brecht}
\affiliation{Paderborn University, Integrated Quantum Optics,\looseness=-1{ }Institute for Photonic Quantum Systems (PhoQS), Warburger Stra\ss{}e 100, 33098 Paderborn, Germany}

\author{Christine Silberhorn}
\affiliation{Paderborn University, Integrated Quantum Optics,\looseness=-1{ }Institute for Photonic Quantum Systems (PhoQS), Warburger Stra\ss{}e 100, 33098 Paderborn, Germany}

\author{Otfried G\"uhne}
\affiliation{Naturwissenschaftlich-Technische Fakult\"at, Universit\"at Siegen, Walter-Flex-Stra\ss{}e 3, 57068 Siegen, Germany}

\date{\today}

\begin{abstract}
    Distributed quantum information in networks is paramount for global secure quantum communication. 
    Moreover, it finds applications as a resource for relevant tasks, such as clock synchronization, magnetic field sensing, and blind quantum computation.
    For quantum network analysis and benchmarking of implementations, however, it is crucial to characterize the topology of networks in a way that reveals the nodes between which entanglement can be reliably distributed.
    Here, we demonstrate an efficient scheme for this topology certification.
    Our scheme allows for distinguishing, in a scalable manner, different networks consisting of bipartite and multipartite entanglement sources. It can be 
    applied to semi-device independent scenarios also, where the measurement 
    devices and network nodes are not well characterized and trusted. We 
    experimentally demonstrate our approach by certifying the topology of 
    different six-qubit networks generated with polarized photons, employing 
    active feed-forward and time multiplexing. 
    Our methods can be used for general simultaneous tests of multiple hypotheses with few measurements, being useful for other certification scenarios in quantum technologies. 
\end{abstract}

\maketitle

{\it Introduction.---}
    A key hallmark of modern information technology is the ability to establish multi-user communication channels.
    In the field of quantum information processing, such multilateral communication channels gain further significance as they can be enhanced by providing entanglement as a quantum resource between multiple nodes of a quantum network \cite{Kimble2008, Wehner2018, Azuma2022}.
    Indeed, such quantum networks can serve as useful structures for secure communication \cite{Murta2020}, clock synchronization \cite{Komar2014}, distributed field sensing \cite{Sekatski2020, Proctor2018}, and even blind quantum computation \cite{Barz2012}.
    Consequently, many experimental groups pursue their implementation by demonstrating basic network 
    structures \cite{Hermans2022,Liu2023, Nitsche2020}. 
    In any case, real quantum networks are fragile and stochastic effects caused by, e.g., probabilistic entanglement generation or failure of nodes and links \cite{Collins2007,Shchukin2019,Weinbrenner2023,Bugalho2023}, detrimentally affect the usefulness and connectivity of a network.
    Similarly, eavesdropping events as well as corrupted nodes may affect the network structure, necessitating a probing and monitoring of the shared quantum resources in an easily accessible manner.
    
\begin{figure}[t]
    \centering
    \includegraphics[clip, trim=1.9cm 7.7cm 2.1cm 10.9cm,width=0.91\columnwidth]{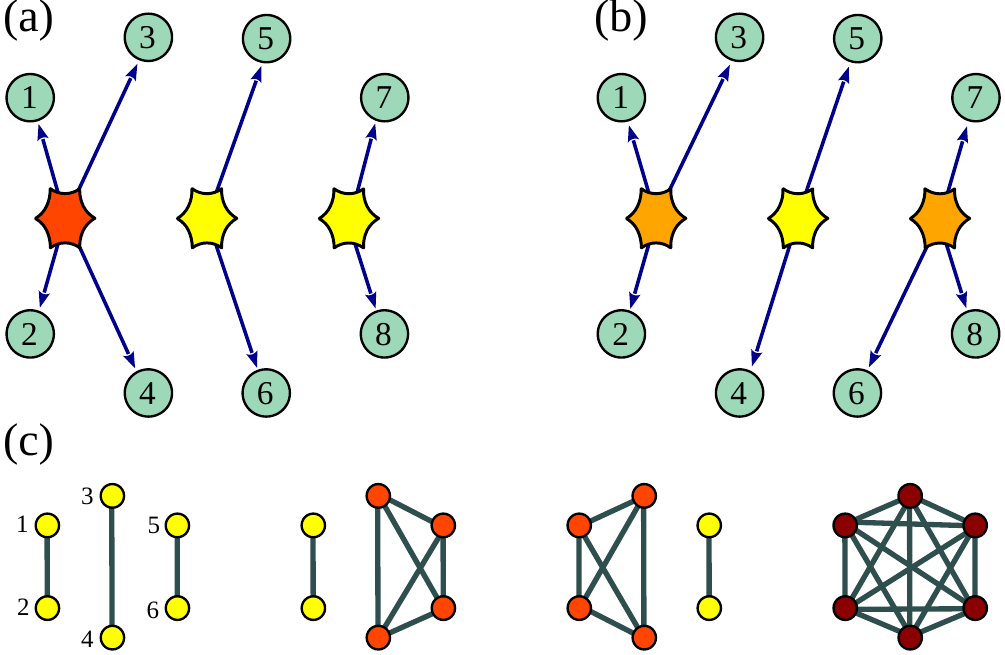}
    \caption{%
        (a,b) For a quantum network of eight parties, two possible network configurations are shown:
        (a) two two-qubit sources and one four-qubit source and (b) one two-qubit source and two three-qubit sources are used to distribute eight qubits. 
        (c) For the experimental implementation, we consider a six-qubit network, where four different configurations shall be distinguished.
        Here, the sources distribute two-, four- or six-qubit GHZ states, which are depicted by the fully connected graphs.
    }\label{fig:problemconf}
\end{figure}

    In any case, the characterization of the distributed entanglement across the network is indispensable.
    So far, however, the tools for this 
    have been limited mainly to analyzing which quantum states
    can and cannot be established in a given network structure \cite{navascues2020genuine, kraft2021quantum, luo2021new, Tavakoli2022, Hansenne2022, Pappa2012, Mccutcheon2016}.
    In order to understand the properties and limitations of a given quantum network, however, it is crucial to certify its topology.
    This refers to the probing of a set of targeted quantum network configurations (see Fig.~\ref{fig:problemconf}) and goes beyond the characterization of single distributed quantum states \cite{Pappa2012,Murta2023,Kao2023}.
    First approaches to this problem have recently been given \cite{Chen2023, Yang2022}, but 
    they assume the distribution of pure states or specific noise models and do not allow for certifying the quantum nature of the distributed states. 

    We explore in theory and experiment the resource-efficient hypothesis testing of distinct quantum network configurations.
    For this purpose, we devise and implement a protocol that allows us to statistically certify (or falsify) which hypothesis is consistent with the multipartite quantum state of a network.
    This is significantly different from standard hypothesis testing where only one hypothesis is compared with one null hypothesis \cite{Flammia2011,Yu2022,Saggio2019,Pallister2018,Martinez2021}.
    Importantly, our tests are based on a common set of  local measurements and are easily implementable.
    In the experiment, we generate six-qubit quantum networks with different multipartite entanglement structures from a flexible, engineered source based on time-multiplexing and feed-forward.
    Our measurements then allow us to determine the generated entanglement configuration with high confidence.
    Finally, we present methods for the topology certification of networks which can also be applied if some nodes are not trusted or some measurement devices are not certified.

{\it Statement of the problem.---}
    To start with an example, consider eight parties connected through a quantum network, where entanglement is distributed.
    In Fig.~\ref{fig:problemconf} two different ways for this are depicted as an example:
    In network (a), the entanglement is distributed by two Bell pair sources together with one four-qubit source, whereas network (b) consists of one Bell pair and two three-qubit sources.
    In this proof of principle we consider only two possible topologies for the sake of conciseness and clarity. Nonetheless, our method extends to an arbitrary number of topologies, including topologies with sources made up of a larger number of qubits.
    Here and in the following, we want to certify that the sources distribute (potentially noisy) Greenberger-Horne-Zeilinger (GHZ) states, 
    \begin{equation}
        \ket{\GHZ_n} = \frac{1}{\sqrt{2}} ( \ket{0}^{\otimes n} + \ket{1}^{\otimes n} ),
    \end{equation}
    consisting of $n$ qubits.
    For two qubits, GHZ states are the maximally entangled Bell states.
    For more particles, they are also, in some sense, maximally entangled \cite{Scarani2001}, 
    and form valuable resources for multiparticle cryptography \cite{Murta2020,Christandl2005} and quantum metrology \cite{Toth2014,Giovannetti2004}.
    The key questions are now:
    How can the eight parties identify, in a simple manner, which of the different configurations the network qubits are currently sharing?
    How can they find out which types of sources have been used (e.g., how many qubits were entangled) and between which of their qubits entanglement has successfully been generated?

    These questions, with appropriate modifications, arise in several situations.
    For instance, it may be that the parties are connected via an intricate network of qubits, where the network provider promises to generate maximally entangled states in different configurations.
    In this case, the parties may be interested in verifying the provider's claims with minimum effort.
    Alternatively, consider a network with some dishonest participants.
    Then, some other participants may want to certify that they share an entangled state, while ensuring that this state is not shared with any potentially malicious party.

    In the general case, the problem can be considered for $N$ nodes, 
    corresponding to $N$ qubits from different sources.
    The aim is then to certify the topology of the network from which the qubits originate.
    Here, one may additionally assume that only GHZ states of maximally $M < N$  qubits can be generated by the sources, effectively reducing the set of possible configurations.
    In the following, we present an efficient scheme to measure all fidelities 
    $
        F_I = \tr(\ket{\GHZ_I}\bra{\GHZ_I} \varrho_I)
    $
    for all possible configurations in a unified manner.
    The index $I$ denotes here the set of $|I|=n$ qubits on which the fidelity depends, and the state $\varrho_I$ is the reduced state on the qubits labeled by $I$;
    in Fig.~\ref{fig:problemconf}(a), the set $I$ could be $\{1,2,3,4\}$, $\{5,6\}$ or $\{7,8\}$. 
    This allows us to derive statistically rigorous tests for the different hypotheses about the topology directly from the measurement data. 
    We stress that our approach is fundamentally different from the task of state discrimination for a set of states \cite{Barnett2009} as we are not assuming that the quantum state comes from a fixed collection of states.
    In addition, we do not make any assumptions about the kind of noise.

{\it Simultaneous fidelity estimation and hypothesis testing.---}
    To start, we recall how the fidelity of an $N$-qubit GHZ state can be determined \cite{Guehne2007}. 
    This state can be decomposed into a diagonal term $\mathcal{D}_N$ and an anti-diagonal term $\mathcal{A}_N$, i.e.,
    \begin{align}
        \ketbra{GHZ_N}  = & \frac{1}{2} 
        \big(\ketbra{0}^{\otimes N} +  \ketbra{1}^{\otimes N}
        + \ket{0}\!\!\bra{1}^{\otimes N} \nonumber \\
        & + \ket{1}\!\!\bra{0}^{\otimes N}\big) = \frac{1}{2}(\mathcal{D}_N + \mathcal{A}_N ).
    \end{align}
    The diagonal term can be determined by performing the Pauli measurement $\sigma_z^{\otimes N}$ on 
    all qubits.
    Concerning the anti-diagonal term, it has been shown that it can be written as $\mathcal{A}_N = \nicefrac{1}{N} \sum_{k=0}^{N-1} (-1)^k \mathcal{M}_k^{\otimes N}$, where the observables $\mathcal{M}_k$ are given by measurements in the $x$-$y$ plane of the Bloch sphere, 
    \begin{equation}
        \mathcal{M}_k = 
        \Big[ 
            \cos \left( \frac{k \pi}{N}\right) \sigma_x + \sin\left( \frac{k \pi}{N} \right) \sigma_y
        \Big].
        \label{eq:M_k}
    \end{equation}
    This means that the fidelity of an $N$-qubit GHZ state can be determined by in total $N+1$ local measurements. 
    Note that the measurements $\mathcal{M}_k$ also depend on the number $N$ of qubits.

    The key observation is that the decomposition of $\mathcal{A}_N$ is not unique.
    Indeed, other sets of measurements in the $x$-$y$ plane of the Bloch sphere also allow us to determine $\mathcal{A}_N$, as long as the measurements form a basis in the space of operators spanned by products of $\sigma_x$ and $\sigma_y$, with an even number of $\sigma_y$ \cite{Guehne2007}.
    This paves the way for the simultaneous estimation of several GHZ fidelities:
    From the measurement data of $\mathcal{M}_k^{\otimes N}$, $\mathcal{A}_N$ can be determined using the formula above.
    Furthermore, for any subset of $m<N$ qubits, the expectation values $\mean{\mathcal{M}_k^{\otimes m}}$ can be obtained from the same set of data, which allows for the computation of the fidelity of the $m$-qubit GHZ states with respect to the reduced state on these $m$ particles.
    Explicit formulas for the $m$-qubit fidelities are provided in Appendix~A.

    This leads to the following scheme for testing the topology of a network: 
    For a given $N$, the parties perform the $N+1$ local measurements $\sigma_z^{\otimes N}$ and $\mathcal{M}_k^{\otimes N}$.
    They then use this data to determine the set of fidelities $\{F_I\}$ for each considered network configuration.
    This allows them to identify the actual configuration and, at the same time, to characterize the quality of the sources.
    If the parties know that the sources are at most $M$-partite, then it suffices to perform the $M$ measurements $\mathcal{M}_k^{\otimes N}$ with angles $\nicefrac{k \pi}{M}$ in the $x$-$y$ plane and the measurement $\sigma_z^{\otimes N}$.

    It remains to discuss how to formulate a proper hypothesis test in the space of potential fidelities that can be used to make a decision based on the observed data.
    Here, the task is to formulate a set of mutually exclusive hypotheses which correspond to the different considered topologies, are physically motivated, and, at the same time, allow for a direct estimation of a $p$-value.
    The $p$-value of a hypothesis describes the probability of observing the experimental data given that the hypothesis $H$ holds true, i.e., $p= \Pr[\mathrm{data} \mid H ]$. 
    From a physical point of view, it is important to certify that a working source delivers GHZ states with a fidelity $F>1/2$, as this guarantees the presence of genuine multiparticle entanglement \cite{Guehne2009,Horodecki2009}.
    Moreover, there are intricate dependencies between the fidelities of different GHZ states.
    If a state on $n$ qubits has a high GHZ fidelity, then the reduced state on a subset of $m<n$ qubits has also a non-vanishing fidelity with a GHZ state (with potentially adjusted phases);
    indeed, we have $F_m > F_n/2$ because of the common entries on the diagonal.

    The above considerations motivate the following strategy to formulate exclusive hypotheses in the space of all fidelities. 
    Any topology $T$ is characterized by a set of fidelities $\{ F^T_I \}$ of the included GHZ states.
    The hypothesis corresponding to $T$ is then given by a set of conditions of the type
    \begin{equation}\label{eq:hypothesis}
        F_I^T - \max_{ G \supset I } \Big\{ F_G^T \Big\} > \frac{1}{2},  
    \end{equation}
    where ${G \supset I}$ denotes the relevant supersets of the qubits in the $n$-qubit set ${I}$.
    For instance, in order to distinguish only the distinct topologies (a) and (b) in Fig.~\ref{fig:problemconf}, the hypothesis for configuration (a) should contain the conditions $F_{\{ 1,2,3,4\}} > 1/2$, $F_{\{ 5,6\}} > 1/2$ and $F_{\{7,8\}} - F_{\{6,7,8\}} > 1/2$, and the hypothesis for (b) the conditions $F_{\{1,2,3\}}-F_{\{1,2,3,4\}}>1/2$, $F_{\{4,5\}}>1/2$ and $F_{\{6,7,8\}}>1/2$, rendering these hypotheses mutually exclusive. 
    Taking the differences of the fidelities, e.g., $F_{\{7,8\}} - F_{\{6,7,8\}} > 1/2$, is necessary to distinguish between tripartite entanglement on $\{6,7,8\}$, which leads to high fidelities $F_{\{7,8\}}$ \textit{and} $F_{\{6,7,8\}}$, and bipartite entanglement on $\{7,8\}$, which only results in a high fidelity $F_{\{7,8\}}$.
    Finally, one always has to consider the null hypothesis, where none of the considered topologies can be certified, e.g., when states different from GHZ states have been prepared.
    Given such hypotheses, an upper bound on the $p$-values can directly be calculated from the data as in Refs.~\cite{Flammia2011,Moroder2013}, using large deviation bounds like the Hoeffding inequality \cite{Hoeffding1963}.
    Details can be found in Appendices B and C.

\begin{figure*}[t]
    \centering
    \includegraphics{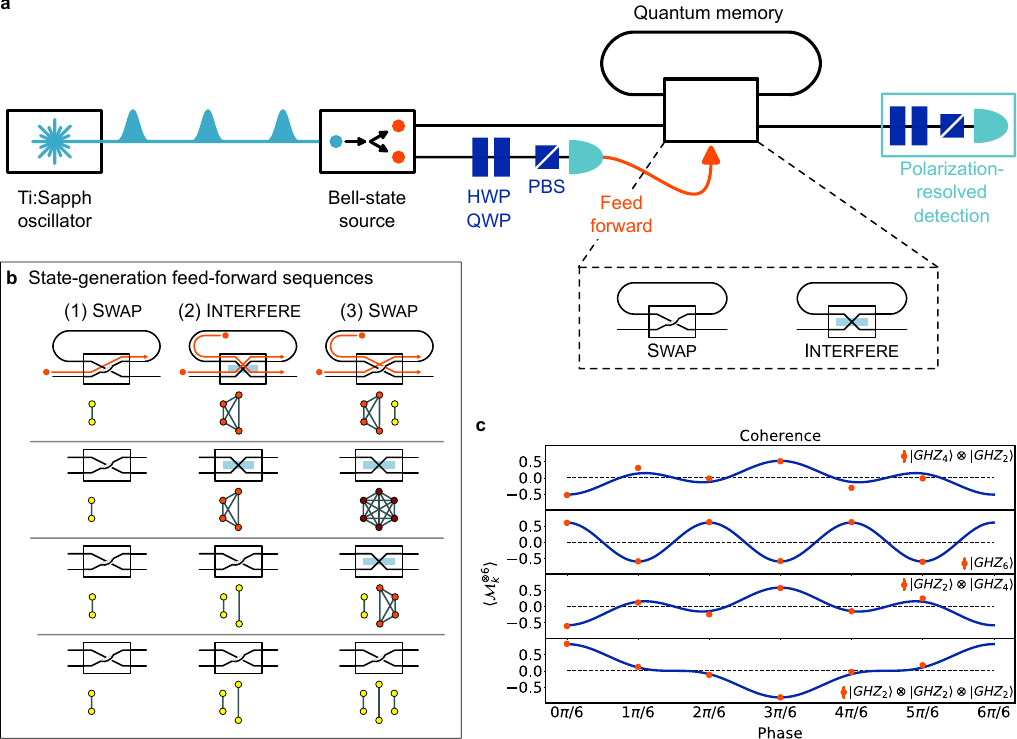}
    \caption{%
        Operation principle of the experiment. 
        (a) A dispersion-engineered, integrated Bell-state source probabilistically generates polarization-entangled Bell states.
        A successful polarization-resolved detection of one pair-photon generates a feed-forward signal to the quantum memory that stores its sibling.
        Retrieved light from the memory is detected in another polarization-resolved detection stage.
        The zoom-out shows the two operation modes of the memory:
        \textsc{swap} stores a new photon while releasing a stored photon without them interacting;
        \textsc{interfere} realizes a balanced interference between a new photon and a stored photon to increase the size of the entangled state.
        HWP: half-wave plate; QWP: quarter-wave plate; PBS: polarizing beam splitter.
        (b) Depending on the desired final network topology, different feed-forward sequences are implemented, which exchange \textsc{swap} and \textsc{interfere} operations of the quantum memory.
        Lines 2-4 show a smaller version of the quantum memory pictogram.
        (c) Measured average coherence terms $\langle\mathcal{M}_k^{\otimes 6}\rangle$ for the four different network topologies (orange markers), in the same order as in (b).
        The terms show an oscillatory dependence on the phase, which correlates with the number of entangled photons in each state.
        The blue lines are theory and serve as guide to the eye.
        Error bars are smaller than the symbol size. 
    }\label{fig:setup}
\end{figure*}

{\it Experimental demonstration and application of the hypothesis tests.---}
    For the experimental implementation we first generate polarization-entangled Bell states with a dispersion-engineered parametric down-conversion source in a periodically poled potassium titanyl waveguide \cite{meyer2018}, see also Fig.~\ref{fig:setup} (a).
    Larger entangled states are created with the help of a polarization qubit memory, based on an all-optical storage loop, which doubles as time-multiplexing device to increase generation rates and beam splitter to interfere successive Bell pairs \cite{Meyer2022}.
    The qubit memory's operation mode---\textsc{swap} or \textsc{interfere}---is triggered by fast feed-forward based on the detection of one qubit from each Bell pair.
    We can generate four- and six-photon GHZ states with our setup at increased rates.
    Here, we make use of the specific programming capabilities of our system to generate the four different six-photon states depicted in Fig.~\ref{fig:problemconf}(c) 
    by changing the feed-forward sequence of the memory, without any physical changes in the experimental setup in practice;
    see Fig.~\ref{fig:setup} (b). 

    The first line schematically depicts the feed-forward sequence for the generation of a $|GHZ_4\rangle\otimes|GHZ_2\rangle$ network topology.
    Upon detection of a photon from the first Bell pair, its partner is stored in the memory by means of a \textsc{swap} operation.
    It is then interfered with a photon from the successive Bell pair to generate a $|GHZ_4\rangle$ state by means of the \textsc{interfere} operation.
    The stored photon is then exchanged for a photon from the third Bell pair via another \textsc{swap} operation.
    Note that we did not depict a final \textsc{swap} operation, which serves to read out the final photon from the memory.
    The \textsc{interfere} operation is essentially a fusion of two polarization qubits from two Bell pairs by interfering them on a polarizing beam splitter and post-selecting on a specific measurement pattern to create a graph state \cite{bouwmeester99,Lu2007}. 
    
    Two consecutive \textsc{interfere} operations generate a $|GHZ_6\rangle$ state, while a \textsc{swap} operation followed by an \textsc{interfere} operation generates the state $|GHZ_2\rangle \otimes |GHZ_4\rangle$.
    Finally, two \textsc{swap} operations yield the state $|GHZ_2\rangle \otimes |GHZ_2\rangle \otimes |GHZ_2\rangle$, where photons from each Bell pair only share entanglement with each other. 
    Note that in our setup fixing the phase of the six-photon GHZ state as $\ket{GHZ_6} = (\ket{0}^{\otimes 6} + \ket{1}^{\otimes 6})/{\sqrt{2}}$ implicitly fixes the phase of the four-photon state to 
    $\ket{GHZ_4^-} = (\ket{0}^{\otimes 4} - \ket{1}^{\otimes 4})/{\sqrt{2}}$, so we formulate the hypotheses for
    this four-photon source.

\begin{figure*}[t]
    \centering
    \includegraphics[width=18.0cm]{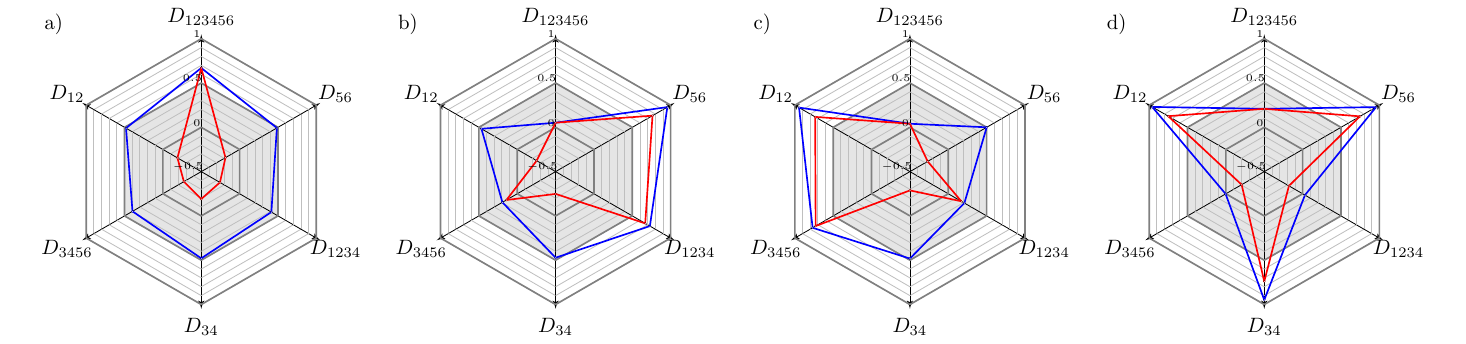}
    \caption{%
        Graphical depiction of the fidelities (blue curves) and the differences of the fidelities considered in the hypotheses (red curves) for the four different measured datasets.
        The different directions $D_I$ denote either the fidelity $F_I$ (blue) or the difference $F_I -\max_{G \supset I} F_G$ (red).  
        The differences allow for a clear separation between the states in the sense that only the desired terms are larger than $1/2$ (outside the dark area).
        The hypothesis test leads to the following results:
        dataset a) belongs to the state $\ket{GHZ_6}$, b) to $\ket{GHZ_4}\otimes \ket{GHZ_2}$, c) to $\ket{GHZ_2}\otimes \ket{GHZ_4}$ and d) to $\ket{GHZ_2}\otimes\ket{GHZ_2}\otimes \ket{GHZ_2}$.
    }\label{fig:radarchart}
\end{figure*}

    Our source generates Bell states with entanglement visibility exceeding 93\%.
    In total, we multiplex up to seven pump pulses to create the three Bell states required for this work.
    This yields a final six-photon event rate of approximately $0.3\,$Hz; more details are given in Appendix~D.
    Note that higher rates can be achieved by multiplexing more pump pulses \cite{Meyer2022}.
    This, however, comes at the cost of a decreased state fidelity.

    As described above, for any network topology we perform the Pauli measurement $\sigma_z^{\otimes 6}$ as well as the measurements of the $\mathcal{M}_k^{\otimes 6}$.
    For the latter, we set the corresponding wave plate angles in front of the detection; see Fig.~\ref{fig:setup}.
    We record around thousand successful events for every measurement setting to ensure good statistics. 
    Our data yields H/V populations of $\mathcal{D}_6 = (74\pm1.7)\%$ and a total coherence value of $\mathcal{A}_6 = (60\pm0.9)\%$, resulting in a total fidelity of $F_6=0.67\pm0.01$ for the $|GHZ_6\rangle$ state \cite{wang16}.
    A plot of all coherence terms for the different topologies is shown in Fig.~\ref{fig:setup}(c).

    We then formulate the hypotheses describing the four different topologies and the additional null hypothesis as described above. 
    In Fig.~\ref{fig:radarchart}, the different measured fidelities for the four different states are depicted together with the differences of the fidelities considered in the hypotheses. 
    For each of the four generated states, we computed the upper bounds on the $p$-values corresponding to the five hypotheses $H_i$ ($i=0,\dots,4$). 
    The upper bounds concerning the null hypothesis are smaller than $9.5 \times 10^{-8}$ for all states, i.e., the probability that the fidelities are too small to certify one of the network states is at most $9.5 \times 10^{-8}$. 
    For the four hypotheses $H_1$, $H_2$, $H_3$ and $H_4$, there exists one hypothesis for each state, for which the $p$-value is trivially upper bounded by one, while the other three $p$-values are upper bounded by at least $1.2 \times 10^{-32}$. Details are given in Appendix~C.

    So far, we assumed that the performed measurements are well calibrated and that the nodes are trusted.
    This, however, may not necessarily be the case.
    The key to discussing the device-independent scenario is to consider the Bell operator from the Mermin inequality \cite{Mermin1990}, which is connected to the GHZ-state fidelity. A detailed discussion of this approach can be found in {Appendix~E}.
    
{\it Discussion.---}
We introduced a method to certify the topology
of $N$-qubit networks with a multiple hypotheses test. 
Notably, only $N+1$ measurements are required, regardless of the number of topologies (i.e., hypotheses) under consideration.
We successfully implemented the method using
different six-photon configurations.
Our work opens an avenue to several new research directions. 
First, our methods can be extended to characterize other network scenarios. For instance, other quantum states
besides GHZ states, such as cluster and graph states \cite{Meignant2019, Hahn2022}, may be distributed and certified. Also, one may include 
the effects of classical communication, imperfect quantum memories, and probabilistic entanglement generation to certify the topology of the 
classical and quantum layer of a network. Second, our approach is an 
example of a multiple hypothesis test, a concept which has natural 
applications to other problems. Examples are the joint estimation 
of several incompatible measurements \cite{Guehne2023}
from simpler ones, or the characterization of different 
quantum gates from the same set of state preparations and 
measurements. Here, also shadow-like techniques based 
on randomized measurements may be a fruitful tool \cite{Huang2020}.


{\it Acknowledgments.---}
We thank Jef Pauwels and Armin Tavakoli for discussions.
This work has been supported by the Deutsche Forschungsgemeinschaft (DFG, German Research Foundation, project numbers 447948357 and 440958198, and the Collaborative Research Center TRR 142 (Project No. 231447078, project C10)), the Sino-German Center for Research Promotion (Project M-0294), and the German Ministry of Education and Research (Project QuKuK, BMBF Grant No. 16KIS1618K). The work was further funded through the Ministerium f\"ur Kultur und Wissenschaft des Landes Nordrhein-Westfalen through the project PhoQC: Photonisches Quantencomputing.
L.T.W., K.H., and S.D. acknowledge support by the House of Young Talents of the University of Siegen.

\vspace{0.5cm}

\bibliography{literature}

\begin{thebibliography}{56}%
\makeatletter
\providecommand \@ifxundefined [1]{%
 \@ifx{#1\undefined}
}%
\providecommand \@ifnum [1]{%
 \ifnum #1\expandafter \@firstoftwo
 \else \expandafter \@secondoftwo
 \fi
}%
\providecommand \@ifx [1]{%
 \ifx #1\expandafter \@firstoftwo
 \else \expandafter \@secondoftwo
 \fi
}%
\providecommand \natexlab [1]{#1}%
\providecommand \enquote  [1]{``#1''}%
\providecommand \bibnamefont  [1]{#1}%
\providecommand \bibfnamefont [1]{#1}%
\providecommand \citenamefont [1]{#1}%
\providecommand \href@noop [0]{\@secondoftwo}%
\providecommand \href [0]{\begingroup \@sanitize@url \@href}%
\providecommand \@href[1]{\@@startlink{#1}\@@href}%
\providecommand \@@href[1]{\endgroup#1\@@endlink}%
\providecommand \@sanitize@url [0]{\catcode `\\12\catcode `\$12\catcode
  `\&12\catcode `\#12\catcode `\^12\catcode `\_12\catcode `\%12\relax}%
\providecommand \@@startlink[1]{}%
\providecommand \@@endlink[0]{}%
\providecommand \url  [0]{\begingroup\@sanitize@url \@url }%
\providecommand \@url [1]{\endgroup\@href {#1}{\urlprefix }}%
\providecommand \urlprefix  [0]{URL }%
\providecommand \Eprint [0]{\href }%
\providecommand \doibase [0]{https://doi.org/}%
\providecommand \selectlanguage [0]{\@gobble}%
\providecommand \bibinfo  [0]{\@secondoftwo}%
\providecommand \bibfield  [0]{\@secondoftwo}%
\providecommand \translation [1]{[#1]}%
\providecommand \BibitemOpen [0]{}%
\providecommand \bibitemStop [0]{}%
\providecommand \bibitemNoStop [0]{.\EOS\space}%
\providecommand \EOS [0]{\spacefactor3000\relax}%
\providecommand \BibitemShut  [1]{\csname bibitem#1\endcsname}%
\let\auto@bib@innerbib\@empty
\bibitem [{\citenamefont {Kimble}(2008)}]{Kimble2008}%
  \BibitemOpen
  \bibfield  {author} {\bibinfo {author} {\bibfnamefont {H.~J.}\ \bibnamefont
  {Kimble}},\ }\bibfield  {title} {\bibinfo {title} {The quantum internet},\
  }\href {https://doi.org/10.1038/nature07127} {\bibfield  {journal} {\bibinfo
  {journal} {Nature}\ }\textbf {\bibinfo {volume} {453}},\ \bibinfo {pages}
  {1023} (\bibinfo {year} {2008})}\BibitemShut {NoStop}%
\bibitem [{\citenamefont {Wehner}\ \emph {et~al.}(2018)\citenamefont {Wehner},
  \citenamefont {Elkouss},\ and\ \citenamefont {Hanson}}]{Wehner2018}%
  \BibitemOpen
  \bibfield  {author} {\bibinfo {author} {\bibfnamefont {S.}~\bibnamefont
  {Wehner}}, \bibinfo {author} {\bibfnamefont {D.}~\bibnamefont {Elkouss}},\
  and\ \bibinfo {author} {\bibfnamefont {R.}~\bibnamefont {Hanson}},\
  }\bibfield  {title} {\bibinfo {title} {Quantum internet: A vision for the
  road ahead},\ }\href {https://doi.org/10.1126/science.aam9288} {\bibfield
  {journal} {\bibinfo  {journal} {Science}\ }\textbf {\bibinfo {volume}
  {362}},\ \bibinfo {pages} {303} (\bibinfo {year} {2018})}\BibitemShut
  {NoStop}%
\bibitem [{\citenamefont {Azuma}\ \emph {et~al.}(2022)\citenamefont {Azuma},
  \citenamefont {Economou}, \citenamefont {Elkouss}, \citenamefont {Hilaire},
  \citenamefont {Jiang}, \citenamefont {Lo},\ and\ \citenamefont
  {Tzitrin}}]{Azuma2022}%
  \BibitemOpen
  \bibfield  {author} {\bibinfo {author} {\bibfnamefont {K.}~\bibnamefont
  {Azuma}}, \bibinfo {author} {\bibfnamefont {S.~E.}\ \bibnamefont {Economou}},
  \bibinfo {author} {\bibfnamefont {D.}~\bibnamefont {Elkouss}}, \bibinfo
  {author} {\bibfnamefont {P.}~\bibnamefont {Hilaire}}, \bibinfo {author}
  {\bibfnamefont {L.}~\bibnamefont {Jiang}}, \bibinfo {author} {\bibfnamefont
  {H.-K.}\ \bibnamefont {Lo}},\ and\ \bibinfo {author} {\bibfnamefont
  {I.}~\bibnamefont {Tzitrin}},\ }\href@noop {} {\bibinfo {title} {Quantum
  repeaters: From quantum networks to the quantum internet}} (\bibinfo {year}
  {2022}),\ \Eprint {https://arxiv.org/abs/2212.10820} {arXiv:2212.10820
  [quant-ph]} \BibitemShut {NoStop}%
\bibitem [{\citenamefont {Murta}\ \emph {et~al.}(2020)\citenamefont {Murta},
  \citenamefont {Grasselli}, \citenamefont {Kampermann},\ and\ \citenamefont
  {Bru{\ss}}}]{Murta2020}%
  \BibitemOpen
  \bibfield  {author} {\bibinfo {author} {\bibfnamefont {G.}~\bibnamefont
  {Murta}}, \bibinfo {author} {\bibfnamefont {F.}~\bibnamefont {Grasselli}},
  \bibinfo {author} {\bibfnamefont {H.}~\bibnamefont {Kampermann}},\ and\
  \bibinfo {author} {\bibfnamefont {D.}~\bibnamefont {Bru{\ss}}},\ }\bibfield
  {title} {\bibinfo {title} {Quantum conference key agreement: A review},\
  }\href {https://doi.org/https://doi.org/10.1002/qute.202000025} {\bibfield
  {journal} {\bibinfo  {journal} {Adv. Quantum Technol.}\ }\textbf {\bibinfo
  {volume} {3}},\ \bibinfo {pages} {2000025} (\bibinfo {year}
  {2020})}\BibitemShut {NoStop}%
\bibitem [{\citenamefont {Komar}\ \emph {et~al.}(2014)\citenamefont {Komar},
  \citenamefont {Kessler}, \citenamefont {Bishof}, \citenamefont {Jiang},
  \citenamefont {S{\o}rensen}, \citenamefont {Ye},\ and\ \citenamefont
  {Lukin}}]{Komar2014}%
  \BibitemOpen
  \bibfield  {author} {\bibinfo {author} {\bibfnamefont {P.}~\bibnamefont
  {Komar}}, \bibinfo {author} {\bibfnamefont {E.~M.}\ \bibnamefont {Kessler}},
  \bibinfo {author} {\bibfnamefont {M.}~\bibnamefont {Bishof}}, \bibinfo
  {author} {\bibfnamefont {L.}~\bibnamefont {Jiang}}, \bibinfo {author}
  {\bibfnamefont {A.~S.}\ \bibnamefont {S{\o}rensen}}, \bibinfo {author}
  {\bibfnamefont {J.}~\bibnamefont {Ye}},\ and\ \bibinfo {author}
  {\bibfnamefont {M.~D.}\ \bibnamefont {Lukin}},\ }\bibfield  {title} {\bibinfo
  {title} {A quantum network of clocks},\ }\href
  {https://doi.org/10.1038/nphys3000} {\bibfield  {journal} {\bibinfo
  {journal} {Nat. Phys.}\ }\textbf {\bibinfo {volume} {10}},\ \bibinfo {pages}
  {582} (\bibinfo {year} {2014})}\BibitemShut {NoStop}%
\bibitem [{\citenamefont {Sekatski}\ \emph {et~al.}(2020)\citenamefont
  {Sekatski}, \citenamefont {W{\"o}lk},\ and\ \citenamefont
  {D{\"u}r}}]{Sekatski2020}%
  \BibitemOpen
  \bibfield  {author} {\bibinfo {author} {\bibfnamefont {P.}~\bibnamefont
  {Sekatski}}, \bibinfo {author} {\bibfnamefont {S.}~\bibnamefont {W{\"o}lk}},\
  and\ \bibinfo {author} {\bibfnamefont {W.}~\bibnamefont {D{\"u}r}},\
  }\bibfield  {title} {\bibinfo {title} {Optimal distributed sensing in noisy
  environments},\ }\href {https://doi.org/10.1103/PhysRevResearch.2.023052}
  {\bibfield  {journal} {\bibinfo  {journal} {Phys. Rev. Res.}\ }\textbf
  {\bibinfo {volume} {2}},\ \bibinfo {pages} {023052} (\bibinfo {year}
  {2020})}\BibitemShut {NoStop}%
\bibitem [{\citenamefont {Proctor}\ \emph {et~al.}(2018)\citenamefont
  {Proctor}, \citenamefont {Knott},\ and\ \citenamefont
  {Dunningham}}]{Proctor2018}%
  \BibitemOpen
  \bibfield  {author} {\bibinfo {author} {\bibfnamefont {T.~J.}\ \bibnamefont
  {Proctor}}, \bibinfo {author} {\bibfnamefont {P.~A.}\ \bibnamefont {Knott}},\
  and\ \bibinfo {author} {\bibfnamefont {J.~A.}\ \bibnamefont {Dunningham}},\
  }\bibfield  {title} {\bibinfo {title} {Multiparameter estimation in networked
  quantum sensors},\ }\href {https://doi.org/10.1103/PhysRevLett.120.080501}
  {\bibfield  {journal} {\bibinfo  {journal} {Phys. Rev. Lett.}\ }\textbf
  {\bibinfo {volume} {120}},\ \bibinfo {pages} {080501} (\bibinfo {year}
  {2018})}\BibitemShut {NoStop}%
\bibitem [{\citenamefont {Barz}\ \emph {et~al.}(2012)\citenamefont {Barz},
  \citenamefont {Kashefi}, \citenamefont {Broadbent}, \citenamefont
  {Fitzsimons}, \citenamefont {Zeilinger},\ and\ \citenamefont
  {Walther}}]{Barz2012}%
  \BibitemOpen
  \bibfield  {author} {\bibinfo {author} {\bibfnamefont {S.}~\bibnamefont
  {Barz}}, \bibinfo {author} {\bibfnamefont {E.}~\bibnamefont {Kashefi}},
  \bibinfo {author} {\bibfnamefont {A.}~\bibnamefont {Broadbent}}, \bibinfo
  {author} {\bibfnamefont {J.~F.}\ \bibnamefont {Fitzsimons}}, \bibinfo
  {author} {\bibfnamefont {A.}~\bibnamefont {Zeilinger}},\ and\ \bibinfo
  {author} {\bibfnamefont {P.}~\bibnamefont {Walther}},\ }\bibfield  {title}
  {\bibinfo {title} {Demonstration of blind quantum computing},\ }\href
  {https://doi.org/10.1126/science.1214707} {\bibfield  {journal} {\bibinfo
  {journal} {Science}\ }\textbf {\bibinfo {volume} {335}},\ \bibinfo {pages}
  {303} (\bibinfo {year} {2012})}\BibitemShut {NoStop}%
\bibitem [{\citenamefont {{Hermans}}\ \emph {et~al.}(2022)\citenamefont
  {{Hermans}}, \citenamefont {{Pompili}}, \citenamefont {{Beukers}},
  \citenamefont {{Baier}}, \citenamefont {{Borregaard}},\ and\ \citenamefont
  {{Hanson}}}]{Hermans2022}%
  \BibitemOpen
  \bibfield  {author} {\bibinfo {author} {\bibfnamefont {S.~L.~N.}\
  \bibnamefont {{Hermans}}}, \bibinfo {author} {\bibfnamefont {M.}~\bibnamefont
  {{Pompili}}}, \bibinfo {author} {\bibfnamefont {H.~K.~C.}\ \bibnamefont
  {{Beukers}}}, \bibinfo {author} {\bibfnamefont {S.}~\bibnamefont {{Baier}}},
  \bibinfo {author} {\bibfnamefont {J.}~\bibnamefont {{Borregaard}}},\ and\
  \bibinfo {author} {\bibfnamefont {R.}~\bibnamefont {{Hanson}}},\ }\bibfield
  {title} {\bibinfo {title} {{Qubit teleportation between non-neighbouring
  nodes in a quantum network}},\ }\href
  {https://doi.org/10.1038/s41586-022-04697-y} {\bibfield  {journal} {\bibinfo
  {journal} {Nature}\ }\textbf {\bibinfo {volume} {605}},\ \bibinfo {pages}
  {663} (\bibinfo {year} {2022})}\BibitemShut {NoStop}%
\bibitem [{\citenamefont {Liu}\ \emph {et~al.}(2023)\citenamefont {Liu},
  \citenamefont {Luo}, \citenamefont {Yu}, \citenamefont {Wang}, \citenamefont
  {Wang}, \citenamefont {Hu}, \citenamefont {Li}, \citenamefont {Zheng},
  \citenamefont {Yao}, \citenamefont {Yan}, \citenamefont {Teng}, \citenamefont
  {Jiang}, \citenamefont {Liu}, \citenamefont {Xie}, \citenamefont {Zhang},
  \citenamefont {Mao}, \citenamefont {Jiang}, \citenamefont {Zhang},
  \citenamefont {Bao},\ and\ \citenamefont {Pan}}]{Liu2023}%
  \BibitemOpen
  \bibfield  {author} {\bibinfo {author} {\bibfnamefont {J.-L.}\ \bibnamefont
  {Liu}}, \bibinfo {author} {\bibfnamefont {X.-Y.}\ \bibnamefont {Luo}},
  \bibinfo {author} {\bibfnamefont {Y.}~\bibnamefont {Yu}}, \bibinfo {author}
  {\bibfnamefont {C.-Y.}\ \bibnamefont {Wang}}, \bibinfo {author}
  {\bibfnamefont {B.}~\bibnamefont {Wang}}, \bibinfo {author} {\bibfnamefont
  {Y.}~\bibnamefont {Hu}}, \bibinfo {author} {\bibfnamefont {J.}~\bibnamefont
  {Li}}, \bibinfo {author} {\bibfnamefont {M.-Y.}\ \bibnamefont {Zheng}},
  \bibinfo {author} {\bibfnamefont {B.}~\bibnamefont {Yao}}, \bibinfo {author}
  {\bibfnamefont {Z.}~\bibnamefont {Yan}}, \bibinfo {author} {\bibfnamefont
  {D.}~\bibnamefont {Teng}}, \bibinfo {author} {\bibfnamefont {J.-W.}\
  \bibnamefont {Jiang}}, \bibinfo {author} {\bibfnamefont {X.-B.}\ \bibnamefont
  {Liu}}, \bibinfo {author} {\bibfnamefont {X.-P.}\ \bibnamefont {Xie}},
  \bibinfo {author} {\bibfnamefont {J.}~\bibnamefont {Zhang}}, \bibinfo
  {author} {\bibfnamefont {Q.-H.}\ \bibnamefont {Mao}}, \bibinfo {author}
  {\bibfnamefont {X.}~\bibnamefont {Jiang}}, \bibinfo {author} {\bibfnamefont
  {Q.}~\bibnamefont {Zhang}}, \bibinfo {author} {\bibfnamefont {X.-H.}\
  \bibnamefont {Bao}},\ and\ \bibinfo {author} {\bibfnamefont {J.-W.}\
  \bibnamefont {Pan}},\ }\href@noop {} {\bibinfo {title} {A multinode quantum
  network over a metropolitan area}} (\bibinfo {year} {2023}),\ \Eprint
  {https://arxiv.org/abs/2309.00221} {arXiv:2309.00221 [quant-ph]} \BibitemShut
  {NoStop}%
\bibitem [{\citenamefont {Nitsche}\ \emph {et~al.}(2020)\citenamefont
  {Nitsche}, \citenamefont {De}, \citenamefont {Barkhofen}, \citenamefont
  {Meyer-Scott}, \citenamefont {Tiedau}, \citenamefont {Sperling},
  \citenamefont {G\'abris}, \citenamefont {Jex},\ and\ \citenamefont
  {Silberhorn}}]{Nitsche2020}%
  \BibitemOpen
  \bibfield  {author} {\bibinfo {author} {\bibfnamefont {T.}~\bibnamefont
  {Nitsche}}, \bibinfo {author} {\bibfnamefont {S.}~\bibnamefont {De}},
  \bibinfo {author} {\bibfnamefont {S.}~\bibnamefont {Barkhofen}}, \bibinfo
  {author} {\bibfnamefont {E.}~\bibnamefont {Meyer-Scott}}, \bibinfo {author}
  {\bibfnamefont {J.}~\bibnamefont {Tiedau}}, \bibinfo {author} {\bibfnamefont
  {J.}~\bibnamefont {Sperling}}, \bibinfo {author} {\bibfnamefont
  {A.}~\bibnamefont {G\'abris}}, \bibinfo {author} {\bibfnamefont
  {I.}~\bibnamefont {Jex}},\ and\ \bibinfo {author} {\bibfnamefont
  {C.}~\bibnamefont {Silberhorn}},\ }\bibfield  {title} {\bibinfo {title}
  {Local versus global two-photon interference in quantum networks},\ }\href
  {https://doi.org/10.1103/PhysRevLett.125.213604} {\bibfield  {journal}
  {\bibinfo  {journal} {Phys. Rev. Lett.}\ }\textbf {\bibinfo {volume} {125}},\
  \bibinfo {pages} {213604} (\bibinfo {year} {2020})}\BibitemShut {NoStop}%
\bibitem [{\citenamefont {Collins}\ \emph {et~al.}(2007)\citenamefont
  {Collins}, \citenamefont {Jenkins}, \citenamefont {Kuzmich},\ and\
  \citenamefont {Kennedy}}]{Collins2007}%
  \BibitemOpen
  \bibfield  {author} {\bibinfo {author} {\bibfnamefont {O.~A.}\ \bibnamefont
  {Collins}}, \bibinfo {author} {\bibfnamefont {S.~D.}\ \bibnamefont
  {Jenkins}}, \bibinfo {author} {\bibfnamefont {A.}~\bibnamefont {Kuzmich}},\
  and\ \bibinfo {author} {\bibfnamefont {T.~A.~B.}\ \bibnamefont {Kennedy}},\
  }\bibfield  {title} {\bibinfo {title} {Multiplexed memory-insensitive quantum
  repeaters},\ }\href {https://doi.org/10.1103/physrevlett.98.060502}
  {\bibfield  {journal} {\bibinfo  {journal} {Phys. Rev. Lett.}\ }\textbf
  {\bibinfo {volume} {98}},\ \bibinfo {pages} {060502} (\bibinfo {year}
  {2007})}\BibitemShut {NoStop}%
\bibitem [{\citenamefont {Shchukin}\ \emph {et~al.}(2019)\citenamefont
  {Shchukin}, \citenamefont {Schmidt},\ and\ \citenamefont {van
  Loock}}]{Shchukin2019}%
  \BibitemOpen
  \bibfield  {author} {\bibinfo {author} {\bibfnamefont {E.}~\bibnamefont
  {Shchukin}}, \bibinfo {author} {\bibfnamefont {F.}~\bibnamefont {Schmidt}},\
  and\ \bibinfo {author} {\bibfnamefont {P.}~\bibnamefont {van Loock}},\
  }\bibfield  {title} {\bibinfo {title} {Waiting time in quantum repeaters with
  probabilistic entanglement swapping},\ }\href
  {https://doi.org/10.1103/physreva.100.032322} {\bibfield  {journal} {\bibinfo
   {journal} {Phys. Rev. A}\ }\textbf {\bibinfo {volume} {100}},\ \bibinfo
  {pages} {032322} (\bibinfo {year} {2019})}\BibitemShut {NoStop}%
\bibitem [{\citenamefont {Weinbrenner}\ \emph {et~al.}(2023)\citenamefont
  {Weinbrenner}, \citenamefont {Vandré}, \citenamefont {Coopmans},\ and\
  \citenamefont {Gühne}}]{Weinbrenner2023}%
  \BibitemOpen
  \bibfield  {author} {\bibinfo {author} {\bibfnamefont {L.~T.}\ \bibnamefont
  {Weinbrenner}}, \bibinfo {author} {\bibfnamefont {L.}~\bibnamefont
  {Vandré}}, \bibinfo {author} {\bibfnamefont {T.}~\bibnamefont {Coopmans}},\
  and\ \bibinfo {author} {\bibfnamefont {O.}~\bibnamefont {Gühne}},\
  }\href@noop {} {\bibinfo {title} {Aging and reliability of quantum networks}}
  (\bibinfo {year} {2023}),\ \Eprint {https://arxiv.org/abs/2305.19976}
  {arXiv:2305.19976 [quant-ph]} \BibitemShut {NoStop}%
\bibitem [{\citenamefont {Bugalho}\ \emph {et~al.}(2023)\citenamefont
  {Bugalho}, \citenamefont {Coutinho}, \citenamefont {Monteiro},\ and\
  \citenamefont {Omar}}]{Bugalho2023}%
  \BibitemOpen
  \bibfield  {author} {\bibinfo {author} {\bibfnamefont {L.}~\bibnamefont
  {Bugalho}}, \bibinfo {author} {\bibfnamefont {B.~C.}\ \bibnamefont
  {Coutinho}}, \bibinfo {author} {\bibfnamefont {F.~A.}\ \bibnamefont
  {Monteiro}},\ and\ \bibinfo {author} {\bibfnamefont {Y.}~\bibnamefont
  {Omar}},\ }\bibfield  {title} {\bibinfo {title} {Distributing multipartite
  entanglement over noisy quantum networks},\ }\href
  {https://doi.org/https://doi.org/10.22331/q-2023-02-09-920} {\bibfield
  {journal} {\bibinfo  {journal} {Quantum}\ }\textbf {\bibinfo {volume} {7}},\
  \bibinfo {pages} {920} (\bibinfo {year} {2023})}\BibitemShut {NoStop}%
\bibitem [{\citenamefont {Navascues}\ \emph {et~al.}(2020)\citenamefont
  {Navascues}, \citenamefont {Wolfe}, \citenamefont {Rosset},\ and\
  \citenamefont {Pozas-Kerstjens}}]{navascues2020genuine}%
  \BibitemOpen
  \bibfield  {author} {\bibinfo {author} {\bibfnamefont {M.}~\bibnamefont
  {Navascues}}, \bibinfo {author} {\bibfnamefont {E.}~\bibnamefont {Wolfe}},
  \bibinfo {author} {\bibfnamefont {D.}~\bibnamefont {Rosset}},\ and\ \bibinfo
  {author} {\bibfnamefont {A.}~\bibnamefont {Pozas-Kerstjens}},\ }\bibfield
  {title} {\bibinfo {title} {Genuine network multipartite entanglement},\
  }\href {https://doi.org/10.1103/PhysRevLett.125.240505} {\bibfield  {journal}
  {\bibinfo  {journal} {Phys. Rev. Lett.}\ }\textbf {\bibinfo {volume} {125}},\
  \bibinfo {pages} {240505} (\bibinfo {year} {2020})}\BibitemShut {NoStop}%
\bibitem [{\citenamefont {Kraft}\ \emph {et~al.}(2021)\citenamefont {Kraft},
  \citenamefont {Designolle}, \citenamefont {Ritz}, \citenamefont {Brunner},
  \citenamefont {G{\"u}hne},\ and\ \citenamefont {Huber}}]{kraft2021quantum}%
  \BibitemOpen
  \bibfield  {author} {\bibinfo {author} {\bibfnamefont {T.}~\bibnamefont
  {Kraft}}, \bibinfo {author} {\bibfnamefont {S.}~\bibnamefont {Designolle}},
  \bibinfo {author} {\bibfnamefont {C.}~\bibnamefont {Ritz}}, \bibinfo {author}
  {\bibfnamefont {N.}~\bibnamefont {Brunner}}, \bibinfo {author} {\bibfnamefont
  {O.}~\bibnamefont {G{\"u}hne}},\ and\ \bibinfo {author} {\bibfnamefont
  {M.}~\bibnamefont {Huber}},\ }\bibfield  {title} {\bibinfo {title} {Quantum
  entanglement in the triangle network},\ }\href
  {https://doi.org/10.1103/PhysRevA.103.L060401} {\bibfield  {journal}
  {\bibinfo  {journal} {Phys. Rev. A}\ }\textbf {\bibinfo {volume} {103}},\
  \bibinfo {pages} {L060401} (\bibinfo {year} {2021})}\BibitemShut {NoStop}%
\bibitem [{\citenamefont {Luo}(2021)}]{luo2021new}%
  \BibitemOpen
  \bibfield  {author} {\bibinfo {author} {\bibfnamefont {M.-X.}\ \bibnamefont
  {Luo}},\ }\bibfield  {title} {\bibinfo {title} {New genuinely multipartite
  entanglement},\ }\href {https://doi.org/10.1002/qute.202000123} {\bibfield
  {journal} {\bibinfo  {journal} {Adv. Quantum Technol.}\ }\textbf {\bibinfo
  {volume} {4}},\ \bibinfo {pages} {2000123} (\bibinfo {year}
  {2021})}\BibitemShut {NoStop}%
\bibitem [{\citenamefont {Tavakoli}\ \emph {et~al.}(2022)\citenamefont
  {Tavakoli}, \citenamefont {Pozas-Kerstjens}, \citenamefont {Luo},\ and\
  \citenamefont {Renou}}]{Tavakoli2022}%
  \BibitemOpen
  \bibfield  {author} {\bibinfo {author} {\bibfnamefont {A.}~\bibnamefont
  {Tavakoli}}, \bibinfo {author} {\bibfnamefont {A.}~\bibnamefont
  {Pozas-Kerstjens}}, \bibinfo {author} {\bibfnamefont {M.-X.}\ \bibnamefont
  {Luo}},\ and\ \bibinfo {author} {\bibfnamefont {M.-O.}\ \bibnamefont
  {Renou}},\ }\bibfield  {title} {\bibinfo {title} {Bell nonlocality in
  networks},\ }\href {https://doi.org/10.1088/1361-6633/ac41bb} {\bibfield
  {journal} {\bibinfo  {journal} {Rep. Prog. Phys.}\ }\textbf {\bibinfo
  {volume} {85}},\ \bibinfo {pages} {056001} (\bibinfo {year}
  {2022})}\BibitemShut {NoStop}%
\bibitem [{\citenamefont {Hansenne}\ \emph {et~al.}(2022)\citenamefont
  {Hansenne}, \citenamefont {Xu}, \citenamefont {Kraft},\ and\ \citenamefont
  {Gühne}}]{Hansenne2022}%
  \BibitemOpen
  \bibfield  {author} {\bibinfo {author} {\bibfnamefont {K.}~\bibnamefont
  {Hansenne}}, \bibinfo {author} {\bibfnamefont {Z.-P.}\ \bibnamefont {Xu}},
  \bibinfo {author} {\bibfnamefont {T.}~\bibnamefont {Kraft}},\ and\ \bibinfo
  {author} {\bibfnamefont {O.}~\bibnamefont {Gühne}},\ }\bibfield  {title}
  {\bibinfo {title} {Symmetries in quantum networks lead to no-go theorems for
  entanglement distribution and to verification techniques},\ }\href
  {https://doi.org/10.1038/s41467-022-28006-3} {\bibfield  {journal} {\bibinfo
  {journal} {Nat. Commun.}\ }\textbf {\bibinfo {volume} {13}},\ \bibinfo
  {pages} {496} (\bibinfo {year} {2022})}\BibitemShut {NoStop}%
\bibitem [{\citenamefont {Pappa}\ \emph {et~al.}(2012)\citenamefont {Pappa},
  \citenamefont {Chailloux}, \citenamefont {Wehner}, \citenamefont {Diamanti},\
  and\ \citenamefont {Kerenidis}}]{Pappa2012}%
  \BibitemOpen
  \bibfield  {author} {\bibinfo {author} {\bibfnamefont {A.}~\bibnamefont
  {Pappa}}, \bibinfo {author} {\bibfnamefont {A.}~\bibnamefont {Chailloux}},
  \bibinfo {author} {\bibfnamefont {S.}~\bibnamefont {Wehner}}, \bibinfo
  {author} {\bibfnamefont {E.}~\bibnamefont {Diamanti}},\ and\ \bibinfo
  {author} {\bibfnamefont {I.}~\bibnamefont {Kerenidis}},\ }\bibfield  {title}
  {\bibinfo {title} {Multipartite entanglement verification resistant against
  dishonest parties},\ }\href {https://doi.org/10.1103/PhysRevLett.108.260502}
  {\bibfield  {journal} {\bibinfo  {journal} {Phys. Rev. Lett.}\ }\textbf
  {\bibinfo {volume} {108}},\ \bibinfo {pages} {260502} (\bibinfo {year}
  {2012})}\BibitemShut {NoStop}%
\bibitem [{\citenamefont {McCutcheon}\ \emph {et~al.}(2016)\citenamefont
  {McCutcheon}, \citenamefont {Pappa}, \citenamefont {Bell}, \citenamefont
  {Mcmillan}, \citenamefont {Chailloux}, \citenamefont {Lawson}, \citenamefont
  {Mafu}, \citenamefont {Markham}, \citenamefont {Diamanti}, \citenamefont
  {Kerenidis} \emph {et~al.}}]{Mccutcheon2016}%
  \BibitemOpen
  \bibfield  {author} {\bibinfo {author} {\bibfnamefont {W.}~\bibnamefont
  {McCutcheon}}, \bibinfo {author} {\bibfnamefont {A.}~\bibnamefont {Pappa}},
  \bibinfo {author} {\bibfnamefont {B.~A.}\ \bibnamefont {Bell}}, \bibinfo
  {author} {\bibfnamefont {A.}~\bibnamefont {Mcmillan}}, \bibinfo {author}
  {\bibfnamefont {A.}~\bibnamefont {Chailloux}}, \bibinfo {author}
  {\bibfnamefont {T.}~\bibnamefont {Lawson}}, \bibinfo {author} {\bibfnamefont
  {M.}~\bibnamefont {Mafu}}, \bibinfo {author} {\bibfnamefont {D.}~\bibnamefont
  {Markham}}, \bibinfo {author} {\bibfnamefont {E.}~\bibnamefont {Diamanti}},
  \bibinfo {author} {\bibfnamefont {I.}~\bibnamefont {Kerenidis}}, \emph
  {et~al.},\ }\bibfield  {title} {\bibinfo {title} {Experimental verification
  of multipartite entanglement in quantum networks},\ }\href
  {https://doi.org/10.1038/ncomms13251} {\bibfield  {journal} {\bibinfo
  {journal} {Nat. Commun.}\ }\textbf {\bibinfo {volume} {7}},\ \bibinfo {pages}
  {13251} (\bibinfo {year} {2016})}\BibitemShut {NoStop}%
\bibitem [{\citenamefont {Murta}\ and\ \citenamefont
  {Baccari}(2023)}]{Murta2023}%
  \BibitemOpen
  \bibfield  {author} {\bibinfo {author} {\bibfnamefont {G.}~\bibnamefont
  {Murta}}\ and\ \bibinfo {author} {\bibfnamefont {F.}~\bibnamefont
  {Baccari}},\ }\href@noop {} {\bibinfo {title} {Self-testing with dishonest
  parties and device-independent entanglement certification in quantum
  networks}} (\bibinfo {year} {2023}),\ \Eprint
  {https://arxiv.org/abs/2305.10587} {arXiv:2305.10587 [quant-ph]} \BibitemShut
  {NoStop}%
\bibitem [{\citenamefont {Kao}\ \emph {et~al.}(2023)\citenamefont {Kao},
  \citenamefont {Huang}, \citenamefont {Tsai}, \citenamefont {Chen},
  \citenamefont {Sun}, \citenamefont {Li}, \citenamefont {Liao}, \citenamefont
  {Chuu}, \citenamefont {Lu},\ and\ \citenamefont {Li}}]{Kao2023}%
  \BibitemOpen
  \bibfield  {author} {\bibinfo {author} {\bibfnamefont {W.-T.}\ \bibnamefont
  {Kao}}, \bibinfo {author} {\bibfnamefont {C.-Y.}\ \bibnamefont {Huang}},
  \bibinfo {author} {\bibfnamefont {T.-J.}\ \bibnamefont {Tsai}}, \bibinfo
  {author} {\bibfnamefont {S.-H.}\ \bibnamefont {Chen}}, \bibinfo {author}
  {\bibfnamefont {S.-Y.}\ \bibnamefont {Sun}}, \bibinfo {author} {\bibfnamefont
  {Y.-C.}\ \bibnamefont {Li}}, \bibinfo {author} {\bibfnamefont {T.-L.}\
  \bibnamefont {Liao}}, \bibinfo {author} {\bibfnamefont {C.-S.}\ \bibnamefont
  {Chuu}}, \bibinfo {author} {\bibfnamefont {H.}~\bibnamefont {Lu}},\ and\
  \bibinfo {author} {\bibfnamefont {C.-M.}\ \bibnamefont {Li}},\ }\href@noop {}
  {\bibinfo {title} {Scalable quantum network determination with
  {E}instein-{P}odolsky-{R}osen steering}} (\bibinfo {year} {2023}),\ \Eprint
  {https://arxiv.org/abs/2303.17771} {arXiv:2303.17771 [quant-ph]} \BibitemShut
  {NoStop}%
\bibitem [{\citenamefont {Chen}\ \emph {et~al.}(2023)\citenamefont {Chen},
  \citenamefont {Doolittle}, \citenamefont {Larson}, \citenamefont {Saleem},\
  and\ \citenamefont {Chitambar}}]{Chen2023}%
  \BibitemOpen
  \bibfield  {author} {\bibinfo {author} {\bibfnamefont {D.~T.}\ \bibnamefont
  {Chen}}, \bibinfo {author} {\bibfnamefont {B.}~\bibnamefont {Doolittle}},
  \bibinfo {author} {\bibfnamefont {J.~M.}\ \bibnamefont {Larson}}, \bibinfo
  {author} {\bibfnamefont {Z.~H.}\ \bibnamefont {Saleem}},\ and\ \bibinfo
  {author} {\bibfnamefont {E.}~\bibnamefont {Chitambar}},\ }\href@noop {}
  {\bibinfo {title} {Inferring quantum network topology using local
  measurements}} (\bibinfo {year} {2023}),\ \Eprint
  {https://arxiv.org/abs/2212.07987} {arXiv:2212.07987 [quant-ph]} \BibitemShut
  {NoStop}%
\bibitem [{\citenamefont {Yang}\ \emph {et~al.}(2022)\citenamefont {Yang},
  \citenamefont {Yang},\ and\ \citenamefont {Luo}}]{Yang2022}%
  \BibitemOpen
  \bibfield  {author} {\bibinfo {author} {\bibfnamefont {X.}~\bibnamefont
  {Yang}}, \bibinfo {author} {\bibfnamefont {Y.-H.}\ \bibnamefont {Yang}},\
  and\ \bibinfo {author} {\bibfnamefont {M.-X.}\ \bibnamefont {Luo}},\
  }\bibfield  {title} {\bibinfo {title} {Strong entanglement distribution of
  quantum networks},\ }\href {https://doi.org/10.1103/PhysRevResearch.4.013153}
  {\bibfield  {journal} {\bibinfo  {journal} {Phys. Rev. Res.}\ }\textbf
  {\bibinfo {volume} {4}},\ \bibinfo {pages} {013153} (\bibinfo {year}
  {2022})}\BibitemShut {NoStop}%
\bibitem [{\citenamefont {Flammia}\ and\ \citenamefont
  {Liu}(2011)}]{Flammia2011}%
  \BibitemOpen
  \bibfield  {author} {\bibinfo {author} {\bibfnamefont {S.~T.}\ \bibnamefont
  {Flammia}}\ and\ \bibinfo {author} {\bibfnamefont {Y.-K.}\ \bibnamefont
  {Liu}},\ }\bibfield  {title} {\bibinfo {title} {Direct fidelity estimation
  from few {P}auli measurements},\ }\href
  {https://doi.org/10.1103/PhysRevLett.106.230501} {\bibfield  {journal}
  {\bibinfo  {journal} {Phys. Rev. Lett.}\ }\textbf {\bibinfo {volume} {106}},\
  \bibinfo {pages} {230501} (\bibinfo {year} {2011})}\BibitemShut {NoStop}%
\bibitem [{\citenamefont {Yu}\ \emph {et~al.}(2022)\citenamefont {Yu},
  \citenamefont {Shang},\ and\ \citenamefont {Gühne}}]{Yu2022}%
  \BibitemOpen
  \bibfield  {author} {\bibinfo {author} {\bibfnamefont {X.-D.}\ \bibnamefont
  {Yu}}, \bibinfo {author} {\bibfnamefont {J.}~\bibnamefont {Shang}},\ and\
  \bibinfo {author} {\bibfnamefont {O.}~\bibnamefont {Gühne}},\ }\bibfield
  {title} {\bibinfo {title} {Statistical methods for quantum state verification
  and fidelity estimation},\ }\href
  {https://doi.org/https://doi.org/10.1002/qute.202100126} {\bibfield
  {journal} {\bibinfo  {journal} {Adv. Quantum Technol.}\ }\textbf {\bibinfo
  {volume} {5}},\ \bibinfo {pages} {2100126} (\bibinfo {year}
  {2022})}\BibitemShut {NoStop}%
\bibitem [{\citenamefont {Saggio}\ \emph {et~al.}(2019)\citenamefont {Saggio},
  \citenamefont {Dimi{\'{c}}}, \citenamefont {Greganti}, \citenamefont
  {Rozema}, \citenamefont {Walther},\ and\ \citenamefont
  {Daki{\'{c}}}}]{Saggio2019}%
  \BibitemOpen
  \bibfield  {author} {\bibinfo {author} {\bibfnamefont {V.}~\bibnamefont
  {Saggio}}, \bibinfo {author} {\bibfnamefont {A.}~\bibnamefont {Dimi{\'{c}}}},
  \bibinfo {author} {\bibfnamefont {C.}~\bibnamefont {Greganti}}, \bibinfo
  {author} {\bibfnamefont {L.~A.}\ \bibnamefont {Rozema}}, \bibinfo {author}
  {\bibfnamefont {P.}~\bibnamefont {Walther}},\ and\ \bibinfo {author}
  {\bibfnamefont {B.}~\bibnamefont {Daki{\'{c}}}},\ }\bibfield  {title}
  {\bibinfo {title} {Experimental few-copy multipartite entanglement
  detection},\ }\href {https://doi.org/10.1038/s41567-019-0550-4} {\bibfield
  {journal} {\bibinfo  {journal} {Nat. Phys.}\ }\textbf {\bibinfo {volume}
  {15}},\ \bibinfo {pages} {935} (\bibinfo {year} {2019})}\BibitemShut
  {NoStop}%
\bibitem [{\citenamefont {Pallister}\ \emph {et~al.}(2018)\citenamefont
  {Pallister}, \citenamefont {Linden},\ and\ \citenamefont
  {Montanaro}}]{Pallister2018}%
  \BibitemOpen
  \bibfield  {author} {\bibinfo {author} {\bibfnamefont {S.}~\bibnamefont
  {Pallister}}, \bibinfo {author} {\bibfnamefont {N.}~\bibnamefont {Linden}},\
  and\ \bibinfo {author} {\bibfnamefont {A.}~\bibnamefont {Montanaro}},\
  }\bibfield  {title} {\bibinfo {title} {Optimal verification of entangled
  states with local measurements},\ }\href
  {https://doi.org/10.1103/PhysRevLett.120.170502} {\bibfield  {journal}
  {\bibinfo  {journal} {Phys. Rev. Lett.}\ }\textbf {\bibinfo {volume} {120}},\
  \bibinfo {pages} {170502} (\bibinfo {year} {2018})}\BibitemShut {NoStop}%
\bibitem [{\citenamefont {Mart\'{\i}nez~Vargas}\ \emph
  {et~al.}(2021)\citenamefont {Mart\'{\i}nez~Vargas}, \citenamefont {Hirche},
  \citenamefont {Sent\'{\i}s}, \citenamefont {Skotiniotis}, \citenamefont
  {Carrizo}, \citenamefont {Mu\~noz Tapia},\ and\ \citenamefont
  {Calsamiglia}}]{Martinez2021}%
  \BibitemOpen
  \bibfield  {author} {\bibinfo {author} {\bibfnamefont {E.}~\bibnamefont
  {Mart\'{\i}nez~Vargas}}, \bibinfo {author} {\bibfnamefont {C.}~\bibnamefont
  {Hirche}}, \bibinfo {author} {\bibfnamefont {G.}~\bibnamefont {Sent\'{\i}s}},
  \bibinfo {author} {\bibfnamefont {M.}~\bibnamefont {Skotiniotis}}, \bibinfo
  {author} {\bibfnamefont {M.}~\bibnamefont {Carrizo}}, \bibinfo {author}
  {\bibfnamefont {R.}~\bibnamefont {Mu\~noz Tapia}},\ and\ \bibinfo {author}
  {\bibfnamefont {J.}~\bibnamefont {Calsamiglia}},\ }\bibfield  {title}
  {\bibinfo {title} {Quantum sequential hypothesis testing},\ }\href
  {https://doi.org/10.1103/PhysRevLett.126.180502} {\bibfield  {journal}
  {\bibinfo  {journal} {Phys. Rev. Lett.}\ }\textbf {\bibinfo {volume} {126}},\
  \bibinfo {pages} {180502} (\bibinfo {year} {2021})}\BibitemShut {NoStop}%
\bibitem [{\citenamefont {Scarani}\ and\ \citenamefont
  {Gisin}(2001)}]{Scarani2001}%
  \BibitemOpen
  \bibfield  {author} {\bibinfo {author} {\bibfnamefont {V.}~\bibnamefont
  {Scarani}}\ and\ \bibinfo {author} {\bibfnamefont {N.}~\bibnamefont
  {Gisin}},\ }\bibfield  {title} {\bibinfo {title} {Spectral decomposition of
  {B}ell's operators for qubits},\ }\href
  {https://doi.org/10.1088/0305-4470/34/30/314} {\bibfield  {journal} {\bibinfo
   {journal} {J. Phys. A Math. Gen.}\ }\textbf {\bibinfo {volume} {34}},\
  \bibinfo {pages} {6043} (\bibinfo {year} {2001})}\BibitemShut {NoStop}%
\bibitem [{\citenamefont {Christandl}\ and\ \citenamefont
  {Wehner}(2005)}]{Christandl2005}%
  \BibitemOpen
  \bibfield  {author} {\bibinfo {author} {\bibfnamefont {M.}~\bibnamefont
  {Christandl}}\ and\ \bibinfo {author} {\bibfnamefont {S.}~\bibnamefont
  {Wehner}},\ }\bibfield  {title} {\bibinfo {title} {Quantum anonymous
  transmissions},\ }in\ \href {https://doi.org/10.1007/11593447_12} {\emph
  {\bibinfo {booktitle} {Lecture Notes in Computer Science}}}\ (\bibinfo
  {publisher} {Springer Berlin Heidelberg},\ \bibinfo {year} {2005})\ pp.\
  \bibinfo {pages} {217--235}\BibitemShut {NoStop}%
\bibitem [{\citenamefont {Tóth}\ and\ \citenamefont
  {Apellaniz}(2014)}]{Toth2014}%
  \BibitemOpen
  \bibfield  {author} {\bibinfo {author} {\bibfnamefont {G.}~\bibnamefont
  {Tóth}}\ and\ \bibinfo {author} {\bibfnamefont {I.}~\bibnamefont
  {Apellaniz}},\ }\bibfield  {title} {\bibinfo {title} {Quantum metrology from
  a quantum information science perspective},\ }\href
  {https://doi.org/10.1088/1751-8113/47/42/424006} {\bibfield  {journal}
  {\bibinfo  {journal} {J. Phys. A Math. Theor.}\ }\textbf {\bibinfo {volume}
  {47}},\ \bibinfo {pages} {424006} (\bibinfo {year} {2014})}\BibitemShut
  {NoStop}%
\bibitem [{\citenamefont {Giovannetti}\ \emph {et~al.}(2004)\citenamefont
  {Giovannetti}, \citenamefont {Lloyd},\ and\ \citenamefont
  {Maccone}}]{Giovannetti2004}%
  \BibitemOpen
  \bibfield  {author} {\bibinfo {author} {\bibfnamefont {V.}~\bibnamefont
  {Giovannetti}}, \bibinfo {author} {\bibfnamefont {S.}~\bibnamefont {Lloyd}},\
  and\ \bibinfo {author} {\bibfnamefont {L.}~\bibnamefont {Maccone}},\
  }\bibfield  {title} {\bibinfo {title} {Quantum-enhanced measurements: Beating
  the standard quantum limit},\ }\href
  {https://doi.org/10.1126/science.1104149} {\bibfield  {journal} {\bibinfo
  {journal} {Science}\ }\textbf {\bibinfo {volume} {306}},\ \bibinfo {pages}
  {1330} (\bibinfo {year} {2004})}\BibitemShut {NoStop}%
\bibitem [{\citenamefont {Barnett}\ and\ \citenamefont
  {Croke}(2009)}]{Barnett2009}%
  \BibitemOpen
  \bibfield  {author} {\bibinfo {author} {\bibfnamefont {S.~M.}\ \bibnamefont
  {Barnett}}\ and\ \bibinfo {author} {\bibfnamefont {S.}~\bibnamefont
  {Croke}},\ }\bibfield  {title} {\bibinfo {title} {Quantum state
  discrimination},\ }\href {https://doi.org/10.1364/AOP.1.000238} {\bibfield
  {journal} {\bibinfo  {journal} {Adv. Opt. Photonics}\ }\textbf {\bibinfo
  {volume} {1}},\ \bibinfo {pages} {238} (\bibinfo {year} {2009})}\BibitemShut
  {NoStop}%
\bibitem [{\citenamefont {G\"uhne}\ \emph {et~al.}(2007)\citenamefont
  {G\"uhne}, \citenamefont {Lu}, \citenamefont {Gao},\ and\ \citenamefont
  {Pan}}]{Guehne2007}%
  \BibitemOpen
  \bibfield  {author} {\bibinfo {author} {\bibfnamefont {O.}~\bibnamefont
  {G\"uhne}}, \bibinfo {author} {\bibfnamefont {C.-Y.}\ \bibnamefont {Lu}},
  \bibinfo {author} {\bibfnamefont {W.-B.}\ \bibnamefont {Gao}},\ and\ \bibinfo
  {author} {\bibfnamefont {J.-W.}\ \bibnamefont {Pan}},\ }\bibfield  {title}
  {\bibinfo {title} {Toolbox for entanglement detection and fidelity
  estimation},\ }\href {https://doi.org/10.1103/PhysRevA.76.030305} {\bibfield
  {journal} {\bibinfo  {journal} {Phys. Rev. A}\ }\textbf {\bibinfo {volume}
  {76}},\ \bibinfo {pages} {030305} (\bibinfo {year} {2007})}\BibitemShut
  {NoStop}%
\bibitem [{\citenamefont {Gühne}\ and\ \citenamefont
  {Tóth}(2009)}]{Guehne2009}%
  \BibitemOpen
  \bibfield  {author} {\bibinfo {author} {\bibfnamefont {O.}~\bibnamefont
  {Gühne}}\ and\ \bibinfo {author} {\bibfnamefont {G.}~\bibnamefont {Tóth}},\
  }\bibfield  {title} {\bibinfo {title} {Entanglement detection},\ }\href
  {https://doi.org/https://doi.org/10.1016/j.physrep.2009.02.004} {\bibfield
  {journal} {\bibinfo  {journal} {Physics Reports}\ }\textbf {\bibinfo {volume}
  {474}},\ \bibinfo {pages} {1} (\bibinfo {year} {2009})}\BibitemShut {NoStop}%
\bibitem [{\citenamefont {Horodecki}\ \emph {et~al.}(2009)\citenamefont
  {Horodecki}, \citenamefont {Horodecki}, \citenamefont {Horodecki},\ and\
  \citenamefont {Horodecki}}]{Horodecki2009}%
  \BibitemOpen
  \bibfield  {author} {\bibinfo {author} {\bibfnamefont {R.}~\bibnamefont
  {Horodecki}}, \bibinfo {author} {\bibfnamefont {P.}~\bibnamefont
  {Horodecki}}, \bibinfo {author} {\bibfnamefont {M.}~\bibnamefont
  {Horodecki}},\ and\ \bibinfo {author} {\bibfnamefont {K.}~\bibnamefont
  {Horodecki}},\ }\bibfield  {title} {\bibinfo {title} {Quantum entanglement},\
  }\href {https://doi.org/10.1103/RevModPhys.81.865} {\bibfield  {journal}
  {\bibinfo  {journal} {Rev. Mod. Phys.}\ }\textbf {\bibinfo {volume} {81}},\
  \bibinfo {pages} {865} (\bibinfo {year} {2009})}\BibitemShut {NoStop}%
\bibitem [{\citenamefont {Moroder}\ \emph {et~al.}(2013)\citenamefont
  {Moroder}, \citenamefont {Kleinmann}, \citenamefont {Schindler},
  \citenamefont {Monz}, \citenamefont {G\"uhne},\ and\ \citenamefont
  {Blatt}}]{Moroder2013}%
  \BibitemOpen
  \bibfield  {author} {\bibinfo {author} {\bibfnamefont {T.}~\bibnamefont
  {Moroder}}, \bibinfo {author} {\bibfnamefont {M.}~\bibnamefont {Kleinmann}},
  \bibinfo {author} {\bibfnamefont {P.}~\bibnamefont {Schindler}}, \bibinfo
  {author} {\bibfnamefont {T.}~\bibnamefont {Monz}}, \bibinfo {author}
  {\bibfnamefont {O.}~\bibnamefont {G\"uhne}},\ and\ \bibinfo {author}
  {\bibfnamefont {R.}~\bibnamefont {Blatt}},\ }\bibfield  {title} {\bibinfo
  {title} {Certifying systematic errors in quantum experiments},\ }\href
  {https://doi.org/10.1103/PhysRevLett.110.180401} {\bibfield  {journal}
  {\bibinfo  {journal} {Phys. Rev. Lett.}\ }\textbf {\bibinfo {volume} {110}},\
  \bibinfo {pages} {180401} (\bibinfo {year} {2013})}\BibitemShut {NoStop}%
\bibitem [{\citenamefont {Hoeffding}(1963)}]{Hoeffding1963}%
  \BibitemOpen
  \bibfield  {author} {\bibinfo {author} {\bibfnamefont {W.}~\bibnamefont
  {Hoeffding}},\ }\bibfield  {title} {\bibinfo {title} {Probability
  inequalities for sums of bounded random variables},\ }\href
  {https://doi.org/10.1080/01621459.1963.10500830} {\bibfield  {journal}
  {\bibinfo  {journal} {J. Am. Stat. Assoc.}\ }\textbf {\bibinfo {volume}
  {58}},\ \bibinfo {pages} {13} (\bibinfo {year} {1963})}\BibitemShut {NoStop}%
\bibitem [{\citenamefont {Meyer-Scott}\ \emph {et~al.}(2018)\citenamefont
  {Meyer-Scott}, \citenamefont {Prasannan}, \citenamefont {Eigner},
  \citenamefont {Quiring}, \citenamefont {Donohue}, \citenamefont {Barkhofen},\
  and\ \citenamefont {Silberhorn}}]{meyer2018}%
  \BibitemOpen
  \bibfield  {author} {\bibinfo {author} {\bibfnamefont {E.}~\bibnamefont
  {Meyer-Scott}}, \bibinfo {author} {\bibfnamefont {N.}~\bibnamefont
  {Prasannan}}, \bibinfo {author} {\bibfnamefont {C.}~\bibnamefont {Eigner}},
  \bibinfo {author} {\bibfnamefont {V.}~\bibnamefont {Quiring}}, \bibinfo
  {author} {\bibfnamefont {J.~M.}\ \bibnamefont {Donohue}}, \bibinfo {author}
  {\bibfnamefont {S.}~\bibnamefont {Barkhofen}},\ and\ \bibinfo {author}
  {\bibfnamefont {C.}~\bibnamefont {Silberhorn}},\ }\bibfield  {title}
  {\bibinfo {title} {High-performance source of spectrally pure, polarization
  entangled photon pairs based on hybrid integrated-bulk optics},\ }\href
  {https://doi.org/10.1364/OE.26.032475} {\bibfield  {journal} {\bibinfo
  {journal} {Opt. Express}\ }\textbf {\bibinfo {volume} {26}},\ \bibinfo
  {pages} {32475} (\bibinfo {year} {2018})}\BibitemShut {NoStop}%
\bibitem [{\citenamefont {Meyer-Scott}\ \emph {et~al.}(2022)\citenamefont
  {Meyer-Scott}, \citenamefont {Prasannan}, \citenamefont {Dhand},
  \citenamefont {Eigner}, \citenamefont {Quiring}, \citenamefont {Barkhofen},
  \citenamefont {Brecht}, \citenamefont {Plenio},\ and\ \citenamefont
  {Silberhorn}}]{Meyer2022}%
  \BibitemOpen
  \bibfield  {author} {\bibinfo {author} {\bibfnamefont {E.}~\bibnamefont
  {Meyer-Scott}}, \bibinfo {author} {\bibfnamefont {N.}~\bibnamefont
  {Prasannan}}, \bibinfo {author} {\bibfnamefont {I.}~\bibnamefont {Dhand}},
  \bibinfo {author} {\bibfnamefont {C.}~\bibnamefont {Eigner}}, \bibinfo
  {author} {\bibfnamefont {V.}~\bibnamefont {Quiring}}, \bibinfo {author}
  {\bibfnamefont {S.}~\bibnamefont {Barkhofen}}, \bibinfo {author}
  {\bibfnamefont {B.}~\bibnamefont {Brecht}}, \bibinfo {author} {\bibfnamefont
  {M.~B.}\ \bibnamefont {Plenio}},\ and\ \bibinfo {author} {\bibfnamefont
  {C.}~\bibnamefont {Silberhorn}},\ }\bibfield  {title} {\bibinfo {title}
  {Scalable generation of multiphoton entangled states by active feed-forward
  and multiplexing},\ }\href {https://doi.org/10.1103/PhysRevLett.129.150501}
  {\bibfield  {journal} {\bibinfo  {journal} {Phys. Rev. Lett.}\ }\textbf
  {\bibinfo {volume} {129}},\ \bibinfo {pages} {150501} (\bibinfo {year}
  {2022})}\BibitemShut {NoStop}%
\bibitem [{\citenamefont {Bouwmeester}\ \emph {et~al.}(1999)\citenamefont
  {Bouwmeester}, \citenamefont {Pan}, \citenamefont {Daniell}, \citenamefont
  {Weinfurter},\ and\ \citenamefont {Zeilinger}}]{bouwmeester99}%
  \BibitemOpen
  \bibfield  {author} {\bibinfo {author} {\bibfnamefont {D.}~\bibnamefont
  {Bouwmeester}}, \bibinfo {author} {\bibfnamefont {J.-W.}\ \bibnamefont
  {Pan}}, \bibinfo {author} {\bibfnamefont {M.}~\bibnamefont {Daniell}},
  \bibinfo {author} {\bibfnamefont {H.}~\bibnamefont {Weinfurter}},\ and\
  \bibinfo {author} {\bibfnamefont {A.}~\bibnamefont {Zeilinger}},\ }\bibfield
  {title} {\bibinfo {title} {Observation of three-photon
  {G}reenberger-{H}orne-{Z}eilinger entanglement},\ }\href
  {https://doi.org/10.1103/PhysRevLett.82.1345} {\bibfield  {journal} {\bibinfo
   {journal} {Phys. Rev. Lett.}\ }\textbf {\bibinfo {volume} {82}},\ \bibinfo
  {pages} {1345} (\bibinfo {year} {1999})}\BibitemShut {NoStop}%
\bibitem [{\citenamefont {Lu}\ \emph {et~al.}(2007)\citenamefont {Lu},
  \citenamefont {Zhou}, \citenamefont {Gühne}, \citenamefont {Gao},
  \citenamefont {Zhang}, \citenamefont {Yuan}, \citenamefont {Goebel},
  \citenamefont {Yang},\ and\ \citenamefont {Pan}}]{Lu2007}%
  \BibitemOpen
  \bibfield  {author} {\bibinfo {author} {\bibfnamefont {C.-Y.}\ \bibnamefont
  {Lu}}, \bibinfo {author} {\bibfnamefont {X.-Q.}\ \bibnamefont {Zhou}},
  \bibinfo {author} {\bibfnamefont {O.}~\bibnamefont {Gühne}}, \bibinfo
  {author} {\bibfnamefont {W.-B.}\ \bibnamefont {Gao}}, \bibinfo {author}
  {\bibfnamefont {J.}~\bibnamefont {Zhang}}, \bibinfo {author} {\bibfnamefont
  {Z.-S.}\ \bibnamefont {Yuan}}, \bibinfo {author} {\bibfnamefont
  {A.}~\bibnamefont {Goebel}}, \bibinfo {author} {\bibfnamefont
  {T.}~\bibnamefont {Yang}},\ and\ \bibinfo {author} {\bibfnamefont {J.-W.}\
  \bibnamefont {Pan}},\ }\bibfield  {title} {\bibinfo {title} {Experimental
  entanglement of six photons in graph states},\ }\href
  {https://doi.org/10.1038/nphys507} {\bibfield  {journal} {\bibinfo  {journal}
  {Nat. Phys.}\ }\textbf {\bibinfo {volume} {3}},\ \bibinfo {pages} {91}
  (\bibinfo {year} {2007})}\BibitemShut {NoStop}%
\bibitem [{\citenamefont {Wang}\ \emph {et~al.}(2016)\citenamefont {Wang},
  \citenamefont {Chen}, \citenamefont {Li}, \citenamefont {Huang},
  \citenamefont {Liu}, \citenamefont {Chen}, \citenamefont {Luo}, \citenamefont
  {Su}, \citenamefont {Wu}, \citenamefont {Li}, \citenamefont {Lu},
  \citenamefont {Hu}, \citenamefont {Jiang}, \citenamefont {Peng},
  \citenamefont {Li}, \citenamefont {Liu}, \citenamefont {Chen}, \citenamefont
  {Lu},\ and\ \citenamefont {Pan}}]{wang16}%
  \BibitemOpen
  \bibfield  {author} {\bibinfo {author} {\bibfnamefont {X.-L.}\ \bibnamefont
  {Wang}}, \bibinfo {author} {\bibfnamefont {L.-K.}\ \bibnamefont {Chen}},
  \bibinfo {author} {\bibfnamefont {W.}~\bibnamefont {Li}}, \bibinfo {author}
  {\bibfnamefont {H.-L.}\ \bibnamefont {Huang}}, \bibinfo {author}
  {\bibfnamefont {C.}~\bibnamefont {Liu}}, \bibinfo {author} {\bibfnamefont
  {C.}~\bibnamefont {Chen}}, \bibinfo {author} {\bibfnamefont {Y.-H.}\
  \bibnamefont {Luo}}, \bibinfo {author} {\bibfnamefont {Z.-E.}\ \bibnamefont
  {Su}}, \bibinfo {author} {\bibfnamefont {D.}~\bibnamefont {Wu}}, \bibinfo
  {author} {\bibfnamefont {Z.-D.}\ \bibnamefont {Li}}, \bibinfo {author}
  {\bibfnamefont {H.}~\bibnamefont {Lu}}, \bibinfo {author} {\bibfnamefont
  {Y.}~\bibnamefont {Hu}}, \bibinfo {author} {\bibfnamefont {X.}~\bibnamefont
  {Jiang}}, \bibinfo {author} {\bibfnamefont {C.-Z.}\ \bibnamefont {Peng}},
  \bibinfo {author} {\bibfnamefont {L.}~\bibnamefont {Li}}, \bibinfo {author}
  {\bibfnamefont {N.-L.}\ \bibnamefont {Liu}}, \bibinfo {author} {\bibfnamefont
  {Y.-A.}\ \bibnamefont {Chen}}, \bibinfo {author} {\bibfnamefont {C.-Y.}\
  \bibnamefont {Lu}},\ and\ \bibinfo {author} {\bibfnamefont {J.-W.}\
  \bibnamefont {Pan}},\ }\bibfield  {title} {\bibinfo {title} {Experimental
  ten-photon entanglement},\ }\href
  {https://doi.org/10.1103/PhysRevLett.117.210502} {\bibfield  {journal}
  {\bibinfo  {journal} {Phys. Rev. Lett.}\ }\textbf {\bibinfo {volume} {117}},\
  \bibinfo {pages} {210502} (\bibinfo {year} {2016})}\BibitemShut {NoStop}%
\bibitem [{\citenamefont {Mermin}(1990)}]{Mermin1990}%
  \BibitemOpen
  \bibfield  {author} {\bibinfo {author} {\bibfnamefont {N.~D.}\ \bibnamefont
  {Mermin}},\ }\bibfield  {title} {\bibinfo {title} {Extreme quantum
  entanglement in a superposition of macroscopically distinct states},\ }\href
  {https://doi.org/10.1103/PhysRevLett.65.1838} {\bibfield  {journal} {\bibinfo
   {journal} {Phys. Rev. Lett.}\ }\textbf {\bibinfo {volume} {65}},\ \bibinfo
  {pages} {1838} (\bibinfo {year} {1990})}\BibitemShut {NoStop}%
\bibitem [{\citenamefont {Meignant}\ \emph {et~al.}(2019)\citenamefont
  {Meignant}, \citenamefont {Markham},\ and\ \citenamefont
  {Grosshans}}]{Meignant2019}%
  \BibitemOpen
  \bibfield  {author} {\bibinfo {author} {\bibfnamefont {C.}~\bibnamefont
  {Meignant}}, \bibinfo {author} {\bibfnamefont {D.}~\bibnamefont {Markham}},\
  and\ \bibinfo {author} {\bibfnamefont {F.}~\bibnamefont {Grosshans}},\
  }\bibfield  {title} {\bibinfo {title} {Distributing graph states over
  arbitrary quantum networks},\ }\href
  {https://doi.org/10.1103/PhysRevA.100.052333} {\bibfield  {journal} {\bibinfo
   {journal} {Phys. Rev. A}\ }\textbf {\bibinfo {volume} {100}},\ \bibinfo
  {pages} {052333} (\bibinfo {year} {2019})}\BibitemShut {NoStop}%
\bibitem [{\citenamefont {Hahn}\ \emph {et~al.}(2022)\citenamefont {Hahn},
  \citenamefont {Dahlberg}, \citenamefont {Eisert},\ and\ \citenamefont
  {Pappa}}]{Hahn2022}%
  \BibitemOpen
  \bibfield  {author} {\bibinfo {author} {\bibfnamefont {F.}~\bibnamefont
  {Hahn}}, \bibinfo {author} {\bibfnamefont {A.}~\bibnamefont {Dahlberg}},
  \bibinfo {author} {\bibfnamefont {J.}~\bibnamefont {Eisert}},\ and\ \bibinfo
  {author} {\bibfnamefont {A.}~\bibnamefont {Pappa}},\ }\bibfield  {title}
  {\bibinfo {title} {Limitations of nearest-neighbor quantum networks},\ }\href
  {https://doi.org/10.1103/PhysRevA.106.L010401} {\bibfield  {journal}
  {\bibinfo  {journal} {Phys. Rev. A}\ }\textbf {\bibinfo {volume} {106}},\
  \bibinfo {pages} {L010401} (\bibinfo {year} {2022})}\BibitemShut {NoStop}%
\bibitem [{\citenamefont {G\"uhne}\ \emph {et~al.}(2023)\citenamefont
  {G\"uhne}, \citenamefont {Haapasalo}, \citenamefont {Kraft}, \citenamefont
  {Pellonp\"a\"a},\ and\ \citenamefont {Uola}}]{Guehne2023}%
  \BibitemOpen
  \bibfield  {author} {\bibinfo {author} {\bibfnamefont {O.}~\bibnamefont
  {G\"uhne}}, \bibinfo {author} {\bibfnamefont {E.}~\bibnamefont {Haapasalo}},
  \bibinfo {author} {\bibfnamefont {T.}~\bibnamefont {Kraft}}, \bibinfo
  {author} {\bibfnamefont {J.-P.}\ \bibnamefont {Pellonp\"a\"a}},\ and\
  \bibinfo {author} {\bibfnamefont {R.}~\bibnamefont {Uola}},\ }\bibfield
  {title} {\bibinfo {title} {Colloquium: Incompatible measurements in quantum
  information science},\ }\href {https://doi.org/10.1103/RevModPhys.95.011003}
  {\bibfield  {journal} {\bibinfo  {journal} {Rev. Mod. Phys.}\ }\textbf
  {\bibinfo {volume} {95}},\ \bibinfo {pages} {011003} (\bibinfo {year}
  {2023})}\BibitemShut {NoStop}%
\bibitem [{\citenamefont {Huang}\ \emph {et~al.}(2020)\citenamefont {Huang},
  \citenamefont {Kueng},\ and\ \citenamefont {Preskill}}]{Huang2020}%
  \BibitemOpen
  \bibfield  {author} {\bibinfo {author} {\bibfnamefont {H.-Y.}\ \bibnamefont
  {Huang}}, \bibinfo {author} {\bibfnamefont {R.}~\bibnamefont {Kueng}},\ and\
  \bibinfo {author} {\bibfnamefont {J.}~\bibnamefont {Preskill}},\ }\bibfield
  {title} {\bibinfo {title} {Predicting many properties of a quantum system
  from very few measurements},\ }\href
  {https://doi.org/10.1038/s41567-020-0932-7} {\bibfield  {journal} {\bibinfo
  {journal} {Nat. Phys.}\ }\textbf {\bibinfo {volume} {16}},\ \bibinfo {pages}
  {1050} (\bibinfo {year} {2020})}\BibitemShut {NoStop}%
\bibitem [{\citenamefont {Zhong}\ \emph {et~al.}(2018)\citenamefont {Zhong},
  \citenamefont {Li}, \citenamefont {Li}, \citenamefont {Peng}, \citenamefont
  {Su}, \citenamefont {Hu}, \citenamefont {He}, \citenamefont {Ding},
  \citenamefont {Zhang}, \citenamefont {Li} \emph {et~al.}}]{zhong2018}%
  \BibitemOpen
  \bibfield  {author} {\bibinfo {author} {\bibfnamefont {H.-S.}\ \bibnamefont
  {Zhong}}, \bibinfo {author} {\bibfnamefont {Y.}~\bibnamefont {Li}}, \bibinfo
  {author} {\bibfnamefont {W.}~\bibnamefont {Li}}, \bibinfo {author}
  {\bibfnamefont {L.-C.}\ \bibnamefont {Peng}}, \bibinfo {author}
  {\bibfnamefont {Z.-E.}\ \bibnamefont {Su}}, \bibinfo {author} {\bibfnamefont
  {Y.}~\bibnamefont {Hu}}, \bibinfo {author} {\bibfnamefont {Y.-M.}\
  \bibnamefont {He}}, \bibinfo {author} {\bibfnamefont {X.}~\bibnamefont
  {Ding}}, \bibinfo {author} {\bibfnamefont {W.}~\bibnamefont {Zhang}},
  \bibinfo {author} {\bibfnamefont {H.}~\bibnamefont {Li}}, \emph {et~al.},\
  }\bibfield  {title} {\bibinfo {title} {12-photon entanglement and scalable
  scattershot boson sampling with optimal entangled-photon pairs from
  parametric down-conversion},\ }\href
  {https://doi.org/10.1103/PhysRevLett.121.250505} {\bibfield  {journal}
  {\bibinfo  {journal} {Physical review letters}\ }\textbf {\bibinfo {volume}
  {121}},\ \bibinfo {pages} {250505} (\bibinfo {year} {2018})}\BibitemShut
  {NoStop}%
\bibitem [{\citenamefont {Kaniewski}(2016)}]{Kaniewski2009}%
  \BibitemOpen
  \bibfield  {author} {\bibinfo {author} {\bibfnamefont {J.}~\bibnamefont
  {Kaniewski}},\ }\bibfield  {title} {\bibinfo {title} {Analytic and nearly
  optimal self-testing bounds for the {C}lauser-{H}orne-{S}himony-{H}olt and
  {M}ermin inequalities},\ }\href
  {https://doi.org/10.1103/PhysRevLett.117.070402} {\bibfield  {journal}
  {\bibinfo  {journal} {Phys. Rev. Lett.}\ }\textbf {\bibinfo {volume} {117}},\
  \bibinfo {pages} {070402} (\bibinfo {year} {2016})}\BibitemShut {NoStop}%
\bibitem [{\citenamefont {Cao}\ \emph {et~al.}(2023)\citenamefont {Cao},
  \citenamefont {Morelli}, \citenamefont {Rozema}, \citenamefont {Zhang},
  \citenamefont {Tavakoli},\ and\ \citenamefont {Walther}}]{Cao2023}%
  \BibitemOpen
  \bibfield  {author} {\bibinfo {author} {\bibfnamefont {H.}~\bibnamefont
  {Cao}}, \bibinfo {author} {\bibfnamefont {S.}~\bibnamefont {Morelli}},
  \bibinfo {author} {\bibfnamefont {L.~A.}\ \bibnamefont {Rozema}}, \bibinfo
  {author} {\bibfnamefont {C.}~\bibnamefont {Zhang}}, \bibinfo {author}
  {\bibfnamefont {A.}~\bibnamefont {Tavakoli}},\ and\ \bibinfo {author}
  {\bibfnamefont {P.}~\bibnamefont {Walther}},\ }\href@noop {} {\bibinfo
  {title} {Genuine multipartite entanglement without fully controllable
  measurements}} (\bibinfo {year} {2023}),\ \Eprint
  {https://arxiv.org/abs/2310.11946} {arXiv:2310.11946 [quant-ph]} \BibitemShut
  {NoStop}%
\bibitem [{\citenamefont {Cotler}\ and\ \citenamefont
  {Wilczek}(2020)}]{cotler2020quantum}%
  \BibitemOpen
  \bibfield  {author} {\bibinfo {author} {\bibfnamefont {J.}~\bibnamefont
  {Cotler}}\ and\ \bibinfo {author} {\bibfnamefont {F.}~\bibnamefont
  {Wilczek}},\ }\bibfield  {title} {\bibinfo {title} {Quantum overlapping
  tomography},\ }\href {https://doi.org/10.1103/PhysRevLett.124.100401}
  {\bibfield  {journal} {\bibinfo  {journal} {Phys. Rev. Lett.}\ }\textbf
  {\bibinfo {volume} {124}},\ \bibinfo {pages} {100401} (\bibinfo {year}
  {2020})}\BibitemShut {NoStop}%
\bibitem [{\citenamefont {Bardyn}\ \emph {et~al.}(2009)\citenamefont {Bardyn},
  \citenamefont {Liew}, \citenamefont {Massar}, \citenamefont {McKague},\ and\
  \citenamefont {Scarani}}]{Bardyn2009}%
  \BibitemOpen
  \bibfield  {author} {\bibinfo {author} {\bibfnamefont {C.-E.}\ \bibnamefont
  {Bardyn}}, \bibinfo {author} {\bibfnamefont {T.~C.~H.}\ \bibnamefont {Liew}},
  \bibinfo {author} {\bibfnamefont {S.}~\bibnamefont {Massar}}, \bibinfo
  {author} {\bibfnamefont {M.}~\bibnamefont {McKague}},\ and\ \bibinfo {author}
  {\bibfnamefont {V.}~\bibnamefont {Scarani}},\ }\bibfield  {title} {\bibinfo
  {title} {Device-independent state estimation based on {B}ell's
  inequalities},\ }\href {https://doi.org/10.1103/PhysRevA.80.062327}
  {\bibfield  {journal} {\bibinfo  {journal} {Phys. Rev. A}\ }\textbf {\bibinfo
  {volume} {80}},\ \bibinfo {pages} {062327} (\bibinfo {year}
  {2009})}\BibitemShut {NoStop}%
\end{thebibliography}%

\onecolumngrid

\newpage

\title{Supplemental Material: \\Certifying the Topology of Quantum Networks: Theory and Experiment}
\date{\today}

\maketitle
\onecolumngrid

In this Supplemental Material we give more details on the theoretical and experimental results of the paper. In Appendix~A, we show in detail how one can retrieve all fidelities for smaller GHZ states from the $N+1$ global measurements described in the main text. In Appendix~B we then apply the Hoeffding inequality, which is a large deviation bound, to the specific problem of fidelity estimation. These results are then used in Appendix~C, in which we first formulate the hypotheses for the experimentally prepared states and then apply the results from Appendix~B to calculate bounds on the $p$-values of all hypotheses. Appendix~D contains more details about the experimental implementation as well as about the obtained data. At last, we show in Appendix~E how one could estimate the fidelity also in the semi-device-independent way, using the Mermin inequality.

\section{Appendix A: Coefficients for the fidelity calculation}
%
    Given an $N$-qubit \GHZ-state, it is known as mentioned in the main text that its fidelity $F_N = \Tr{\varrho \ket{\GHZ_N}\bra{\GHZ_N} }$ with some state $\varrho$ can be obtained by measuring $N+1$ different measurement settings only \cite{Guehne2007}. These consist of the Pauli-$Z$ measurement
    $
        \sigma_z^{\otimes N}
    $
    and the $N$ measurements
    \begin{align}
        \mathcal{M}_k^{\otimes N} = \left[ \cos \left( \frac{k \pi}{N}\right) \sigma_x + \sin\left( \frac{k \pi}{N} \right) \sigma_y \right]^{\otimes N},
    \end{align} 
    with $k=0,\dots, N-1$. Since these are all local measurements, all  measurement results of $\sigma_z^{\otimes n}$ and $\mathcal{M}_k^{\otimes n}$ for $n\leq N$ can be deduced from the same measurement data.
    We now show that, by using the results of these measurements, we can also compute the fidelity of any $n$-qubit reduced state of $\varrho$ (denoted $\varrho^{(n)}$) with an $n$-party \GHZ-state for $n < N$.
    For simplicity, we just write $F_n$ for this fidelity and omit the notation on which $n$ parties the fidelity is calculated.
    
    First, we note that the $\GHZ_n$ state can be written as 
    \begin{align}
        \ket{\GHZ_n}\bra{\GHZ_n} &= \frac{1}{2} \left( \ket{0}\bra{0}^{\otimes n} + \ket{1}\bra{1}^{\otimes n} \right) \\
        &+ \frac{1}{2} \left( \ket{0}\bra{1}^{\otimes n} + \ket{1}\bra{0}^{\otimes n} \right) \\
        &=: \frac{1}{2} ( \mathcal{D}_n + \mathcal{A}_n ) \label{eq:A_n}
    \end{align}
    such that the fidelity $F_n$ reads as
    \begin{align}
        F_n &= \Tr (\varrho \ket{\GHZ_n}\bra{\GHZ_n} \otimes \id_{N-n}) \\
        &= \frac{1}{2}\left[ \Tr (\varrho^{(n)} \mathcal{D}_n) + \Tr (\varrho^{(n)}\mathcal{A}_n) \right].
    \end{align}
    The expectation values of $\ket{0}\bra{0}^{\otimes n}$ and $\ket{1}\bra{1}^{\otimes n}$, and thus of $\mathcal{D}_n$, can be directly recovered from the measurement results of $\sigma_z^{\otimes N}$ for all $n\leq N$. 
    The calculation of the expectation value of $\mathcal{A}_n$ is not as trivial, and we show that there exist real coefficients $a_k$ ($k=0, \dots, N-1$) such that 
    \begin{align} \label{eq:DefOfVeca}
        \sum_{k=0}^{N-1} a_k \mathcal{M}_k^{\otimes n} = \mathcal{A}_n.
    \end{align}
    This means that the expectation value of $\mathcal{A}_n$ can be obtained from the measurement results of $\mathcal{M}_k^{\otimes n}$. Note that the fidelity of $\varrho$ with the $\ket{GHZ_n^-}$ state, which is given by 
    $
        \ket{\GHZ_n^-} = \frac{1}{\sqrt{2}} ( \ket{0}^{\otimes n} - \ket{1}^{\otimes n} ), 
    $
    can also be calculated directly from the same data by just switching the sign from $\mathcal{A}_n$ to $-\mathcal{A}_n$.
    
    \begin{theorem}\label{thm:fourier}
        The $N$-dimensional vector of coefficients $\vec{a}$ in Eq.~\eqref{eq:DefOfVeca} is given by
        $
            \vec{a} = \diag(e^{-i\pi \frac{kn}{N}}) \mathcal{F}^{-1}(\vec{c}),
        $
        where $\diag(e^{-i\pi \frac{kn}{N}})$ denotes the diagonal matrix with entries $e^{-i\pi \frac{kn}{N}}$ ($k=0,\dots, N-1$) and $\mathcal{F}^{-1}$ denotes the inverse discrete Fourier transform (DFT) of $\vec{c} = (1,0,\dots,0,1,c_{n+1},\dots,c_{N-1}) \in\mathbb{C}^N $.
    \end{theorem}
    \begin{proof}
        Let us first reformulate $\mathcal{M}_k$ as 
        \begin{align*}
            \mathcal{M}_k 
            = e^{-i\frac{k\pi}{N}}\ket{0}\bra{1} + e^{i\frac{k\pi}{N}} \ket{1}\bra{0}, 
        \end{align*}
        for $k=0,\dots, N-1$, by rewriting $\sigma_x$ and $\sigma_y$ in the computational basis. We therefore have
        \begin{align*}
            \sum_{k=0}^{N-1} a_k \mathcal{M}_k^{\otimes n} 
            &= \sum_{j=0}^n \left[ \sum_{k=0}^{N-1} e^{-i\pi\frac{k}{N} j} e^{i\pi\frac{k}{N} (n-j)} a_k \right] \\
            &\qquad \cdot \sum_\pi \left[ \bigotimes_{l=0}^j \ket{0}\bra{1} \bigotimes_{l=j+1}^n \ket{1}\bra{0} \right] \\
            &=: \sum_{j=0}^n c_j \sum_\pi \left[ \bigotimes_{l=0}^j \ket{0}\bra{1} \bigotimes_{l=j+1}^n \ket{1}\bra{0} \right], 
        \end{align*}
        where $\sum_\pi \dots$ denotes the sum over all permutations leading to different terms.
        It now follows directly that we must have $c_0 = c_n = 1$ and $c_{1} = \dots = c_{n-1} = 0$ with
        \begin{align}\label{eq:coeffFourier}
            c_j 
            = \sum_{k=0}^{N-1} e^{-2\pi i\frac{k j}{N}} e^{i\pi\frac{kn}{N}} a_k 
        \end{align}
        for $j=0,\dots,n$.
        These are the first $n+1$ entries of the DFT of $\diag(e^{i\pi\frac{kn}{N}}) \vec{a}$.
        By defining $c_{n+1},\dots,c_{N-1}$ through Eq.~\eqref{eq:coeffFourier}, we can consider the extended vector $\vec{c} = (1,0,\dots,0,1,c_{n+1},\dots,c_{N-1})$ as the entire DFT $\vec{c} = \mathcal{F}(\diag(e^{i\pi\frac{kn}{N}}) \vec{a})$. 
        Using the inverse of the DFT, we arrive at
        $
            \vec{a} = \diag(e^{- i\pi\frac{kn}{N}}) \mathcal{F}^{-1}(\vec{c}),
        $
        or, alternatively,
        $
            a_k = \frac{1}{N} e^{- i\pi\frac{kn}{N}} \sum_{j=0}^{N-1} e^{2\pi i\frac{k j}{N}}  c_j,
        $
        with $c_{n+1},\dots,c_{N-1} \in \mathbb{C}$.
    \end{proof}
    
    Note that the coefficients $a_k$ are not necessarily real for an arbitrary choice of $c_{n+1},\dots,c_{N-1}$.
    In that case, one can simply take the real part $\mathrm{Re}(\vec{a})$ only without affecting the validity of Eq.~\eqref{eq:DefOfVeca} as $\mathcal{M}_k$ and $\mathcal{A}_n$ are hermitian.
    However, we show in the following Corollary that, from requiring a minimal norm of the coefficients $\norm{\vec{a}}_2$, it already follows that the vector $\vec{a}$ is real-valued.
    Minimizing the norm $\norm{\vec{a}}_2$ leads to the best bounds on the accuracy of our fidelity estimate. More details about this bound can be found in Appendix~B.
    
    For the next part, we recall that $\norm{\mathcal{F}(\vec{a})}_2^2 = N \norm{\vec{a}}_2^2 $ holds true in the case of the DFT.
    Since $\vec{c}$ is the DFT of $\vec{a}$, up to a complex phase, it directly follows that $\norm{\vec{c}}_2^2 = N \norm{\vec{a}}_2^2$.
    Keeping this in mind, we can directly calculate the coefficients $a_k$ with minimal norm $\norm{\vec{a}}_2$ from the above Theorem.
    \begin{corollary}\label{cor:coefficients}
        The coefficients $a_k$ minimizing the norm $\norm{\vec{a}}_2$ for fixed $n$ and $N$ are given by
        \begin{align}
            a_k = \frac{1}{N} \begin{cases}
            (-1)^k\ &\text{for } n=N, \\
            2\cos(\pi\frac{kn}{N})\ &\text{for } n < N .
            \end{cases}
        \end{align}
    \end{corollary}
    \begin{proof}
        Since the vectors $\vec{a}$ and $\vec{c}$ fulfill $\norm{\vec{c}}_2^2 = N \norm{\vec{a}}_2^2$, the norm of $\vec{a}$ is minimal if and only if the norm of $\vec{c}$ is minimal.
        According to Theorem \ref{thm:fourier}, the norm of the vector $\vec{c}$ reads as
        $
            \norm{\vec{c}}_2^2 = 1 + 0 + \dots + 0 + 1 + |c_{n+1}|^2 + \dots + |c_{N-1}|^2 ,
        $
        which is minimal for $c_{n+1} = \dots = c_{N-1} =0$.
        Thus, the ideal choice for the vector $\vec{c}$ only contains up to two entries equal to one and has therefore a norm of $\norm{\vec{c}}_2 = \sqrt{2}$ for $n<N$ and $\norm{\vec{c}}_2 = 1$ for $n=N$.
        In the case of $n<N$, this leads to 
        $
            N a_k 
            = e^{- i\pi\frac{kn}{N}} + e^{\pi i\frac{k n}{N}} 
            = 2 \cos(\pi\frac{kn}{N})
        $
        and, in the case of $n=N$, to 
        $
            N a_k = e^{- i\pi\frac{kN}{N}} = (-1)^k,
        $
        which proves the statement. 
    \end{proof}
    Note that in the case $n=N$ these are exactly the coefficients given in the main text before. 
    It follows directly from the above Corollary that the norm of the minimal coefficient vector $\vec{a}$ only depends on the number of qubits $\norm{\vec{a}}_2^2 = \frac{1}{N} \norm{c}_2^2 \leq \frac{2}{N}$. 
    We use this result in Appendix~B to characterize the accuracy of the fidelities obtained from the experimental data.
    
\section{Appendix B: Large deviation bounds for fidelity estimation}
    
    In the asymptotic limit, the experimental data for a measurement $M$ ought to reproduce exactly the expectation value $\langle M \rangle$.
    However, since every experiment is restricted to a finite number of measurements, one can only calculate estimates of the expectation values from the measurement results.
    Here we describe a way to bound the accuracy of the calculated estimate using Hoeffding's inequality \cite{Flammia2011, Hoeffding1963}. 
    
    The setting is the following:
    We want to calculate the fidelity $F$, which is a linear combination of some observables $M_i$,
    \begin{align}
        F = \sum_{i=0}^N d_i \langle M_i \rangle .
    \end{align}
    In our case, the coefficients and measurements will be the same as described Appendix~A.
    For now, however, we will look at the general case described above.
    In the experiment, each of the measurements $M_i$ will be measured for a finite number $m_i$ where each of these single measurements leads to a measurement result $A_{ij}$ ($j=1,\dots, m_i$). For example, measuring $m_i=10$ times the measurement $M_i=\sigma_z$ leads to the ten measurement results $A_{ij}=+1$ or $A_{ij}=-1$ for $j=1,\dots, 10$.
    The estimate for the expectation value of $M_i$ is then calculated by
    \begin{align}
        \widehat{ \langle M_i \rangle} = \frac{1}{m_i} \sum_{j=1}^{m_i} A_{ij},
    \end{align}
    and the estimate for the fidelity $F$ is obtained accordingly by
    \begin{align}
        \hat{F}  = \sum_{i=0}^N d_i \widehat{ \langle M_i \rangle} = \sum_{i=0}^N \frac{ d_i}{m_i} \sum_{j=1}^{m_i} A_{ij} .
    \end{align}
    We now want to characterize the accuracy of this estimate by using Hoeffding's inequality \cite{Hoeffding1963}, which we briefly recall:
    \begin{lemma}[Hoeffding's inequality \cite{Hoeffding1963}]
        Let $X_1, \dots, X_m$ be independent, bounded random variables such that there exist $s_i$ and $t_i$ with $s_i \leq X_i \leq t_i$.
        Then, the sum $S_n = \sum_{i=1}^n X_i$ fulfills
        \begin{align}
            \Pr[ (S_n - \langle S_n ) \rangle > \epsilon] < \exp{-\frac{2 \epsilon}{C}},
        \end{align}
        with $C= \sum_{i=1}^m (t_i - s_i)^2$.
    \end{lemma}
    For a proof, see \cite{Hoeffding1963}.
    
    Using this, we can now calculate the accuracy of the estimate $\hat{F}$.
    
    \begin{lemma}
        Defining the fidelities, coefficients, and measurements in the same way as in Appendix~A, and denoting the minimal number of measurements by $\mu= \min_i (m_i)$, the estimate of the fidelity $\esf$ obeys 
        $
            \Pr[ (\esf_n - F_n ) > \epsilon] < \delta
        $
        for $\delta(\epsilon) = \exp{- \frac{8 \mu \epsilon^2}{1+ 8/N}} \Leftrightarrow \epsilon(\delta) = \sqrt{-\frac{1+8/N}{8 \mu} \ln{\delta}}$.
    \end{lemma}
    \begin{proof}
        Following the notation in Appendix~A, the fidelity reads
        \begin{align}
            F_n = \sum_{k=0}^{N-1} \frac{a_k}{2} \langle \mathcal{M}_k^{\otimes n} \rangle + \frac{1}{2} \langle \mathcal{D}_n \rangle =: \sum_{k=0}^{N} \frac{a_k}{2} \langle M_k \rangle, 
        \end{align}
        for $M_k = \mathcal{M}_k^{\otimes n}$ for $k=0,\dots, N-1$, $M_N = \mathcal{D}_n$ and $a_N=1$.
        Note that each measurement $M_k$ has only two possible measurement results:
        the measurements $\mathcal{M}_k^{\otimes n}$ can only result in either $+1$ or $-1$, and the measurement $\mathcal{D}_n$ only in either $0$ or $1$.
        Denoting the number of times each measurement $M_i$ is performed by $m_i$, and the result of the $j$-th measurement of $M_i$ by $A_{ij}$, the estimate of the fidelity is then
        \begin{align}
            \esf_n = \sum_{i=0}^N \frac{a_i}{2 m_i} \sum_{j=1}^{m_i} A_{ij}.
        \end{align}
        Note that this is a sum of the independent random variables $\frac{a_i}{2 m_i} A_{ij}$, whose expectation value yields the true fidelity,
        \begin{align}
            \langle \esf_n \rangle =  \sum_{i=0}^N \frac{a_i}{2 m_i} \sum_{j=1}^{m_i} \langle A_{ij} \rangle
            = \sum_{i=0}^N \frac{a_i}{2 m_i} \sum_{j=1}^{m_i} \langle M_i \rangle
            = \sum_{i=0}^N \frac{a_i}{2} \langle M_i \rangle = F_n.
        \end{align}
        The random variables $A_{ij}$ themselves can only take two different values;
        for $i=N$, it holds
        \begin{align}
            \frac{a_N}{2 m_N} A_{Nj} \in \left\lbrace \frac{a_N}{2 m_N}\times 0, \frac{a_N}{2 m_N}\times 1 \right\rbrace = \left\lbrace 0, \frac{1}{2 m_N} \right\rbrace 
        \end{align}
        and, for $i=0,\dots, N-1$,
        \begin{align}
            \frac{a_i}{2 m_i} A_{ij} \in \left\lbrace \frac{a_i}{2 m_i}\times (-1), \frac{a_i}{2 m_i}\times 1 \right\rbrace.
        \end{align}
        Using Hoeffding's inequality, we therefore arrive at
        \begin{align}
            \Pr[ \left( \esf_n - F_n \right) > \epsilon ] = \Pr( \esf_n - \langle \esf_n \rangle > \epsilon ) \leq \exp{ - \frac{2 \epsilon^2 }{C} }
        \end{align}
        with
        \begin{align}
            C &= \sum_{i=0}^{N-1} \sum_{j=1}^{m_i} \left( \frac{a_i}{2 m_i} - \frac{- a_i}{2 m_i} \right)^2 + \sum_{j=1}^{m_N} \left( \frac{1}{2 m_N} \right)^2 \\
            &= \sum_{i=0}^{N-1} \frac{a_i^2}{m_i} + \frac{1}{4 m_N} \\
            &\leq \frac{1}{4 \mu} \left( 4 \norm{\vec{a}}_2^2 + 1 \right) \\
            &\leq \frac{1}{4 \mu} \left( 8/N + 1 \right),
        \end{align}
        which proves the statement.
    \end{proof}
    The last inequality follows directly from Corollary \ref{cor:coefficients} and explains the need to minimize the norm of the coefficient vector $\vec{a}$.
    
    In Appendix~C, we want to test hypotheses which compare fidelities on different subsets.
    Let us for now denote two of these different fidelities by $F_1$ and $F_2$.
    Then, we can bound in exactly the same way as described above the probability by 
    \begin{align}\label{eq:estFidelity}
        \Pr[(\esf_1 - \esf_2 ) - (F_1 - F_2) > \epsilon] < \exp{-\frac{2 \epsilon^2 \mu}{1+8/N}}.
    \end{align} 
    Intuitively, the factor of $4$ in the constant $C$ appears because the length of the intervals in which the summed random variables assume their values approximately doubles. 
    Mathematically, the random variables $\frac{a^{(1)}_i}{2 m_i} A^{(1)}_{ij} - \frac{a^{(2)}_i}{2 m_i} A^{(2)}_{ij}$ arising in the calculation of the fidelity can now take four values, which are bounded by sums of the different coefficients $a^{(k)}_i$.
    Using the subadditivity of the norm $\norm{\cdot}_2$ in the calculation of the constant $C$ yields then the expression given above.
    
\section{Appendix C: Hypothesis testing}
    
    Recall that our goal is to distinguish between different topologies of distributed states. Using the calculated fidelities for subsets of $n$ of the $N$ qubits, we develop here a hypothesis test to decide which layout describes the measured data best. We do that by first formulating the different topologies as partitions of the set of all qubits. Then, we derive a hypothesis test for different partitions and calculate the $p$-value of each hypothesis using the results from Appendix~B.
    
    We start by noticing that the different layouts of a network can be seen as different partitions of the set $\{1, \dots, N\}$.
    We recall that a partition $P$ of $\{1, \dots, N\}$ is a set of subsets $\{ I_1,\dots,I_k \}$ of $\{1,\dots, N\}$ such that $I_i\neq \emptyset$ and $I_i \cap I_j = \emptyset$ for all $i,j$ and the union of all subsets covers again $\{1,\dots, N\}$,
    \begin{align}
        \bigcup_{i=1}^k I_i = \{1,\dots, N\} .
    \end{align}
    For instance, the configurations (a) and (b) from Fig.~\ref{fig:problemconf} would correspond to the partitions $P_a = \{\{1,2,3,4\},\{5,6\},\{7,8\}\}$ and $P_b = \{\{1,2,3\}, \{4,5\},\{6,7,8\}\}$, respectively.
    
    As seen in the section before, all the fidelities $F_n = \Tr{\varrho \ket{\GHZ_n}\bra{\GHZ_n} \otimes \id_{N-n} }$ for $n\leq N$ can be computed from only $N+1$ measurements.
    We now refine the notation to keep track on which $n$ parties the fidelity is calculated.
    We denote the $n$-party \GHZ-state on the $n$ qubits $I \subseteq \{1,\dots, N\}$, $|I|=n$, by $\ket{\GHZ_I}$ and the respective fidelity by $F_I = \Tr{\varrho \ket{\GHZ_I}\bra{\GHZ_I} \otimes \id_{\overline{I}}}$.
    The overline $\overline{I}$ denotes the complement of $I$ in $\{1,\dots, N\}$. 
    
    Now, let $\mathcal{P}(\{1,\dots, N\})$ be the set of all partitions of $\{1,\dots, N\}$. 
    Physically, we interpret a given partition $P = \{ I_1,\dots,I_k \}$ as the subsets of parties $I_i$ on which a \GHZ-state was distributed.
    This means we expect the global state of the network to be
    \begin{align}
        \bigotimes_{i=1}^k \ket{\GHZ_{I_i}}\bra{\GHZ_{I_i}}  = \bigotimes_{I\in P} \ket{\GHZ_I}\bra{\GHZ_I} .
    \end{align}
    In our case, we have $N=6$ photons which can be entangled in four different ways: 
    \begin{subequations}
        \begin{align}
            \ket{\Psi_1} =& \ket{GHZ_6} , \label{eq:Psi_1}\\
            \ket{\Psi_2} =& \ket{GHZ_4} \otimes \ket{GHZ_2},\\
            \ket{\Psi_3} =& \ket{GHZ_2} \otimes \ket{GHZ_4}, \quad \mathrm{and}\\
            \ket{\Psi_4} =& \ket{GHZ_2} \otimes \ket{GHZ_2} \otimes \ket{GHZ_2}.
        \end{align}
    \end{subequations}
    
    \noindent These four cases correspond to the four partitions 
    \begin{subequations}\label{eq:partitions}
        \begin{align}
            P_1 =& \{ \{ 1,2,3,4,5,6 \} \} ,\\
            P_2 =& \{ \{ 1,2,3,4 \} , \{ 5,6 \} \}  ,\\
            P_3 =& \{ \{ 1,2 \} , \{ 3,4,5,6 \} \} , \quad \mathrm{and}\\
            P_4 =& \{ \{ 1,2 \} , \{ 3,4 \} , \{ 5,6 \} \},
        \end{align}
    \end{subequations}
    respectively.
    To discriminate these different possible configurations from the measured data, we now develop a hypothesis test in which the hypotheses exclude each other, ensuring that only one hypothesis can be accepted. 
    
    \begin{lemma}\label{lem:hypotheses}
        Let $P_1, \dots, P_K \in \mathcal{P}(\{1,\dots, N\})$ be $K$ different partitions such that for every pair of partitions $P_i$ and $P_j$, there exists at least one subset in each partition $\tilde{I}\in P_i$ and $\tilde{J} \in P_j$ with $\tilde{I} \subsetneq \tilde{J}$ or $\tilde{I} \supsetneq \tilde{J}$.
        Then, the $K+1$ hypotheses $H_1, \dots, H_K$ and $H_\emptyset$ given by
        \begin{align}
            H_j \ &: \ \forall I\in P_j \ : \ F_I - \max_{G\supset I} F_G > \frac{1}{2} \quad \forall j\in \{1,\dots,K \} \label{eq:H_i}, \\
            H_\emptyset \ &: \ \mathrm{otherwise}
        \end{align}
        are pairwise exclusive to each other, i.e. only one of them can be accepted.
    \end{lemma}
    \begin{proof}
        Clearly, the different hypotheses $H_i$ and the hypothesis $H_\emptyset$ are mutually exclusive by definition.
        So, it is only left to show that any two hypotheses $H_i$ and $H_j$, with $i\neq j$, exclude each other.
        From the assumptions, there exist two sets $\tilde{I} \in P_i$ and $\tilde{J} \in P_j$ which fulfill without loss of generality $\tilde{I} \subsetneq \tilde{J}$.
        Let us now assume that both hypotheses could be accepted at the same time.
        It follows that
        \begin{align}
            F_{\tilde{I}} &> \frac{1}{2} + \max_{G\supset \tilde{I}} F_G \geq  \frac{1}{2} +  F_{\tilde{J}} \quad \mathrm{and} \\
            F_{\tilde{J}} &> \frac{1}{2} + \max_{G\supset \tilde{J}} F_G \geq  \frac{1}{2} .
        \end{align}
        But then we have
        \begin{align}
            F_{\tilde{I}} > \frac{1}{2} +  F_{\tilde{J}} > \frac{1}{2} + \frac{1}{2} = 1 ,
        \end{align}
        which contradicts the fact that fidelities can be at most one, i.e. $F_{\tilde{I}} \leq 1$, and completes the proof.
    \end{proof}
    
    The different partitions of Eq.~\eqref{eq:partitions} obviously fulfill the condition in Lemma \ref{lem:hypotheses}.
    So, in our case, we have the five following different hypotheses:
    \begin{subequations} \label{eq:thehyps}
        \begin{align}
            H_1 \ &: \  F_{123456} > \nicefrac{1}{2} ; \\
            H_2 \ &: \  \begin{cases}
                h_{21} \ &: \ F_{1234} - F_{123456} > \nicefrac{1}{2} \quad\mathrm{and}\\
                h_{22} \ &: \ F_{56} - \max\{ F_{3456}, F_{123456} \} > \nicefrac{1}{2} ; \\
            \end{cases} \\
            H_3 \ &: \  \begin{cases}
                h_{31} \ &: \ F_{12} - \max\{ F_{1234}, F_{123456} \} > \nicefrac{1}{2} \quad\mathrm{and}\\
                h_{32} \ &: \ F_{3456} - F_{123456} > \nicefrac{1}{2} ;
            \end{cases} \\ 
            H_4 \ &: \  \begin{cases}
                h_{41} \ &: \ F_{12} - \max\{ F_{1234}, F_{123456} \} > \nicefrac{1}{2} , \\
                h_{42} \ &: \ F_{34} - \max\{ F_{1234}, F_{3456}, F_{123456} \} > \nicefrac{1}{2} \quad\mathrm{and}\\
                h_{43} \ &: \ F_{56} - \max\{ F_{3456}, F_{123456} \} > \nicefrac{1}{2},  \\
            \end{cases} \\
             \quad \text{and} \notag  \\
            H_\emptyset \ &: \ \text{otherwise},
        \end{align}
    \end{subequations}
    respectively. Note that, for the sake of simplicity, we wrote, e.g., $12$ for the set $\{ 1,2\}$. 
    
    So, if we have the true fidelities $F_I$ on every one of these specific subsets, we can unambiguously decide which hypothesis or configuration is true and reject the other possibilities.
    However, as we only have access to finite statistics, and therefore to estimates of the fidelities, we now have to calculate the $p$-values of every single hypothesis.
    The $p$-value describes the probability to get a certain experimental result, in our case the calculated estimate of the fidelity, given that a hypothesis is true. 
    
    Since the hypotheses $H_j$ of Eq.~\eqref{eq:H_i} are all of the same form, we concentrate on a single one of these, consisting of $m$ terms of the type $F_I - \max_{J \supset I} F_J > \frac{1}{2}$.
    For the sake of readability, we neglect the index $I$ in the following calculations and denote the hypothesis by $H$, consisting of the sub-hypotheses $h_i \ : \ F^{(i)} - \max_{J} F^{(i)}_J > \frac{1}{2}$ ($i=1,\dots,k$).
    Additionally, we introduce the shorthand notation $X^{(i)}_J := F^{(i)} - F^{(i)}_J$, or equivalently $\min_{J} X^{(i)}_J := F^{(i)} - \max_{J} F^{(i)}_J$.
    With this notation, the $p$-value of the hypothesis for the minimum of the estimates
    \begin{align}
        d_i := \min_{J} \widehat{X^{(i)}_J} ,
    \end{align} 
    which are calculated from the experimentally observed data, reads as
    \begin{align}\label{eq:pvalue}
        p = \Pr[ \min_{J} \widehat{X^{(i)}_J} \leq d_i \ \forall i \in \{ 1,\dots,k \} \mid H ] .
    \end{align}
    
    First, we prove a short and later helpful Lemma.
    \begin{lemma}\label{lem:probMin}
    With the notation introduced above and the fidelities defined as in Appendix~A and B, it holds
        \begin{align}\label{eq:probMin}
        \Pr[ \left( \min_{J} \widehat{X^{(i)}_J} - \min_{J} X^{(i)}_J \right) > \epsilon ] \leq \exp{-\frac{2 \epsilon^2 \mu}{1+8/N}} .
     \end{align}
    \end{lemma}
    \begin{proof}
        To be able to use the results from Appendix~B, we note that the minimum $\min^{(i)}_{J} X^{(i)}_J =: X^{(i)}_{\tilde{J}}$ is assumed for at least one $\tilde{J}$. Additionally, the probability for the minimum of the estimates $\min_{J} \widehat{X^{(i)}_J}$ being larger than some value is upper bounded by the probability for any of the $\widehat{X^{(i)}_J}$ to be larger than that same value.
        We therefore arrive at
        \begin{align}
            \Pr[ \left( \min_{J} \widehat{X^{(i)}_J} - \min_{J} X^{(i)}_J \right) > \epsilon ] 
            = \Pr[ \left( \min_{J} \widehat{X^{(i)}_J} - X^{(i)}_{\tilde{J}} \right) > \epsilon ]
            \leq \Pr[ \left( \widehat{X^{(i)}_{\tilde{J}}} - X^{(i)}_{\tilde{J}} \right) > \epsilon ]
            \overset{\eqref{eq:estFidelity}}{\leq} \exp{-\frac{2 \epsilon^2 \mu}{1+8/N}} = \delta(\epsilon) ,
        \end{align}
        with $\epsilon = \sqrt{-\frac{1+8/N}{2 \mu} \ln(\delta)}$.
    \end{proof}
    
    Now, since it is not possible to calculate the $p$-value exactly given that it only compares fidelities, we provide an upper bound for the $p$-value.
    
    \begin{lemma}\label{lem:pvalue}
        For the hypothesis $H$ described above, the $p$-value in Eq.~\eqref{eq:pvalue} is bounded by
        \begin{align}
            p \leq \delta(d)
        \end{align}
        where $d$ and $\delta(d)$ are defined by $d = \min \{ d_1,\dots, d_k, 1/2 \}$ and $\delta(d) = \exp{-\frac{2 (1/2-d)^2 \mu}{1+8/N}}$.
    \end{lemma}
    \begin{proof}
        We use the fact that the random variables $\min_J X_J^{(i)}$ and their estimates are resulting from quantum states $\varrho$.
        We define regions in the state space where the different hypotheses are fulfilled,
        \begin{align}
            R_i :=& \{ \varrho \mid h_i \ : \ \min_J X_J^{(i)} > \nicefrac{1}{2} \ \mathrm{is\ true}\}, \quad i\in \{ 1,\dots, k \} \quad \mathrm{and} \\
            R :=& \{ \varrho \mid h_i \ : \ \min_J X_J^{(i)} > \nicefrac{1}{2} \ \mathrm{is\ true} \ \forall i \in \{ 1,\dots, k \} \}.
        \end{align}
        Obviously, it holds $R \subseteq R_i$.
        We now can calculate a bound on the $p$-value by maximizing over the possible quantum states which lie in $R$:
        \begin{align}
            p &= \Pr[ \min_{J} \widehat{X^{(j)}_J} \leq d_j \ \forall j \mid H ]
            \leq \max_{\varrho \in R} \Pr[ \min_{J} \widehat{X^{(j)}_J} \leq d_j \ \forall j \mid \varrho ]
            \leq \max_{\varrho \in R_i} \Pr[ \min_{J} \widehat{X^{(j)}_J} \leq d_j \ \forall j \mid \varrho ]  \quad \forall i .
        \end{align}
        Since this holds for every $R_i$, it also holds for the minimum over $i$.
        Additionally, the joint probability for the different events $\min_{J} \widehat{X^{(j)}_J} \leq d_j$ is bounded by the probability of one of these events.
        We therefore arrive at
        \begin{align}
            p &\leq \min_i \max_{\varrho \in R_i} \Pr[ \min_{J} \widehat{X^{(j)}_J} \leq d_j \ \forall j \mid \varrho ]
            \leq \min_i \max_{\varrho \in R_i} \Pr[ \min_{J} \widehat{X^{(i)}_J} \leq d_i \mid \varrho ] .
        \end{align}
        However, the probability in the last line can be bounded using Lemma \ref{lem:probMin}.
        Indeed, for the case $d_i < 1/2$, it holds
        \begin{align}
            \max_{\varrho \in R_i} \Pr[ \min_{J} \widehat{X^{(i)}_J} \leq d_i \mid \varrho ] 
            &\leq \Pr[ \left( \min_J X_J^{(i)} - \min_{J} \widehat{X^{(i)}_J} \right) \geq 1/2 - d_i ]  \\
            &\overset{\eqref{eq:probMin}}{\leq} \exp{-\frac{2 (1/2 - d_i)^2 \mu}{1+8/N}} =: \delta(d_i). 
        \end{align}
        For the case $d_i \geq 1/2$, however, we only obtain the trivial bound 
        \begin{align}
            \max_{\varrho \in R_i} \Pr[ \min_{J} \widehat{X^{(i)}_J} \leq d_i \mid \varrho ] 
            \leq 1 = \delta(1/2).
        \end{align}
        Since $\delta(d_i)$ decreases for decreasing $d_i < 1/2$, we get
        \begin{align}
            p &\leq \min_i \max_{\varrho \in R_i} \Pr[ \min_{J} \widehat{X^{(i)}_J} \leq d_i \mid \varrho ]
            \leq \delta( \min \{ d_1,\dots, d_k, 1/2  \} ),
        \end{align}
        which concludes the proof.
    \end{proof}
    The intuition behind this proof is that the probability for the observed data given the hypothesis $H$ is upper bounded by the probability of the least probable event.
    Since the hypothesis $H$ consists of statements of the type $d_i > 1/2$ the least probable result $d_i$ is the one farthest away from $1/2$, namely $\min_i d_{i} < 1/2$. 

    The only hypothesis which cannot be described in the above manner is the null hypothesis $H_\emptyset$ in Lemma \ref{lem:hypotheses}.
    However, this hypothesis can be written as
    \begin{align}
        H_\emptyset \ : \ \forall j \in \{ 1,\dots, K\} \  \exists I^{(j)} \in P_j \ : \ \min_G (F_{I^{(j)}} - F_G) \leq 1/2.
    \end{align}
    This means that for every other hypothesis $H_j$ there exists at least one sub-hypothesis pertaining to the set $I^{(j)}$ which is not fulfilled.
    For a fixed combination of the sets $I^{(j)}$, we can argue in the same way as in Lemma \ref{lem:pvalue}:
    the probability for the calculated estimate for one sub-hypothesis is bounded by
    \begin{align}
        \Pr[ \min_{G} ( \widehat{F_{I^{(j)}}} - \widehat{F_G}) \geq d_{I^{(j)}} \mid \min_G (F_{I^{(j)}} - F_G) \leq 1/2 ] \leq 
        \begin{cases}
            \begin{rcases}
            \delta(d_{I^{(j)}}) &\text{for } d_{I^{(j)}} >1/2,\\
            1 &\text{else},
            \end{rcases}
        \end{cases} \!
        =: \Phi( d_{I^{(j)}} )
    \end{align}
    with $\delta(d_{I^{(j)}}) = \exp{-\frac{2 (1/2-d_{I^{(j)}})^2 \mu}{1+8/N}}$. The function $\Phi(d)$ denotes the map which assigns to $d$ the value $\delta(d)$, if $d > 1/2$, and $1$ otherwise. Thus, we have
    \begin{align}
        \Pr[ \min_{G} (\widehat{F_{I^{(j)}}} - \widehat{F_G}) \geq d_{I^{(j)}} \ \forall j \mid H_\emptyset ] \leq \min_j \Phi(d_{I^{(j)}}) .
    \end{align}
    The $p$-value itself is then bounded by the maximum over all different combinations of the sets $I^{(j)}$ 
    \begin{align}
        p &\leq \max_{\{I^{(j)}\}} \Pr[ \min_{G} (\widehat{F_{I^{(j)}}} - \widehat{F_G}) \geq d_{I^{(j)}} \ \forall j \mid H_{\emptyset} ] \\
        &\leq \max_{\{I^{(j)}\}} \min_j \Phi( d_{I^{(j)}} ) \\
        &= \min_{j \in \{1,\dots, K\} } \max_{i \in \{1,\dots, k\} } \Phi( d^{(j)}_i ),
    \end{align}
    where $d^{(j)}_i$ are the calculated estimates for the sub-hypothesis $h_i$ in the hypothesis $H_j$.

    Intuitively, this can again be understood in the following way: The probability for the observed data given the hypothesis $H_\emptyset$ is upper bounded by the probability of the least probable event. The least probable event given that $H_\emptyset$ is true, is the event where for a hypothesis $H_j$ all statements of the form $d^{(j)}_i > 1/2$ are fulfilled, since $H_\emptyset$ states that at least for one $i$ it should hold $d^{(j)}_i \leq 1/2$. Since the exact $d^{(j)}_i$ one has to consider here is unknown, one can only upper bound the probability by the worst case $\Phi(d^{(j)}_i)$. 

    In Table~\ref{tab:pvalue} the $p$-values obtained from the experimental data for the hypotheses described in Eq.~\ref{eq:hypothesis} are presented. One can clearly see that for every data set the $p$-value of exactly one hypothesis is trivially upper bounded by $1$ while the $p$-values for all other hypotheses are at most bounded by $9.4765\times 10^{-8}$, suggesting to accept the trivially bounded hypothesis.

    At last, we want to consider the scaling of the above described hypothesis testing scheme. First, we find that for $N$ qubits one needs to perform $N+1$ different measurement settings to calculate all GHZ fidelities, so the number of measurement settings increases linearly in $N$. Additionally, the bounds on the $p$-values depend on $N$, the minimal number of measurement repetitions $\mu$ and the fidelity estimates $d_i$ (see Lemma 7), but become independent from the number of qubits considered for large $N$. Assuming that the experimental fidelities and therefore the estimates $d_i$ are comparable, and fixing the number measurement repetitions to $\mu$ for all measurements, we find that the number of required measurement events scales like $(N+1)\mu$. Another point one has to take into account is the number of hypotheses one wants to consider in the hypothesis test, as the number of potentially certifiable hypotheses scales exponentially in the number of qubits $N$. Since one is mostly interested in only a few practically relevant partitions this should not be a problem for at least up to dozens of qubits.

\begin{table}[] 
    \caption{%
        Upper bound on the $p$-value of the different hypotheses for the four different states.
        We recall that the hypotheses are explicitly given in Eq.~\eqref{eq:thehyps}, and that $H_1$ corresponds to $F_{123456}$ being large, $H_2$ to $F_{1234}$ and $F_{56}$ being large, $H_3$ to $F_{12}$ and $F_{3456}$ being large, $H_4$ to $F_{12}$, $F_{34}$ and $F_{56}$ being large and, finally, $H_\emptyset$ is the null hypothesis.
    }\label{tab:pvalue}
    \centering
    \begin{tabular}{c|c|c|c|c|c}
         &  $H_1$  & $H_2$ & $H_3$ & $H_4$ & $H_\emptyset$ \\ \hline \hline
        Dataset 1 & $1$ & $2.3724\times 10^{-139}$ & $9.2153 \times 10^{-146}$ & $3.0256\times 10^{-116}$ & $9.4765 \times 10^{-8}$\\
        Dataset 2 & $8.3274\times 10^{-92}$ & $1$ & $1.6231 \times 10^{-258}$ & $1.6231\times 10^{-258}$ & $1.4347 \times 10^{-13}$\\
        Dataset 3 & $6.1789\times 10^{-93}$ & $6.8294\times 10^{-259}$ & $1$ & $8.8629\times 10^{-272}$ & $1.3010 \times 10^{-23}$\\
        Dataset 4 & $1.1810\times 10^{-32}$ & $3.7960\times 10^{-174}$ & $4.9169\times 10^{-188}$ & $1$ & $2.3295 \times 10^{-22}$
        \\\hline
    \end{tabular}
\end{table}

\section{Appendix D: Additional data and experimental details}

    An in-depth description and characterization of our experiment can be found in 
    Ref. \cite{Meyer2022}. A dispersion-engineered, spectrally decorrelated guided-wave 
    parametric down-conversion in a periodically poled potassium titanyl waveguide is 
    pumped with ultrafast pump pulses with a central wavelength of $\lambda_p=775\,$nm.
    The source is arranged in a Sagnac configuration to generate polarization-entangled Bell 
    states with a central wavelength of $\lambda_\mathrm{Bell}=1550\,$nm \cite{meyer2018}.
    A set of tomographic wave plates (half-wave plate, quarter-wave plate) and a polarizing beam splitter are used in one arm of the source to implement a polarization-resolving heralding measurement.
    We use superconducting nanowire single-photon detectors with an efficiency of more than 70\% and recovery time of $11\,$ns for photon detection.
    These allow to operate the source at the full laser repetition rate of $76\,$MHz, corresponding to a time difference between successive laser pulses of $13.6\,$ns.
    A successful heralding detection event is fed forward to a high-speed electro-optic switch inside the quantum memory.
    The switch is, effectively, a Pockels cell that rotates the polarization of an incoming light pulse when a high voltage is applied.
    It has a response time (rise and/or fall time) of less than $5\,$ns.
    Depending on when the switching occurs, we can realize the \textsc{swap} and \textsc{interfere} operations described in the main text.
    Again, more details are found in Ref. \cite{Meyer2022}.

    The quantum memory itself is an all-optical, free-space storage loop with a memory bandwidth beyond $1\,$THz, a memory efficiency of 91\%, and a lifetime of $131\,$ns.
    Its operation state is controlled by a field-programmable gate array (FPGA), which converts the heralding detection events into switching time lists for the Pockels cell as well as a time gate for the photon detection behind the memory.
    The latency of the feed-forward is compensated for by sending the photons through a $300\,$m long single-mode fiber.
    After retrieval from the memory, photons are sent to another polarization-resolving detection stage.
    A successful datum is registered when six photons are detected in total, three herald photons in different time bins together with three photons from the memory in the corresponding time bins.    Hence, a complete experiment consists of the following steps:
    Firstly, define the network topology and the corresponding switching sequence;
    secondly, perform a test measurement of the H/V populations to assess system performance;
    thirdly, measure all relevant observables by automatically setting the corresponding wave-plate angles and collecting data until a predetermined number of successful events (in our case around thousand) has been measured.

    Figure \ref{fig:sppdata} exemplarily shows the click-correlation matrices for the H/V populations of the four network topologies considered in the main text (top row) and the $\langle \mathcal{M}_3^{\otimes 6} \rangle$ coherence term (lower row).
    Note that the same data sets 
    have been used to calculate the observables as described in the main text. In Table \ref{tab:pvalue} 

\begin{figure*}[!tbp]
    \centering
    \makebox[\textwidth]{\includegraphics[width=.8\paperwidth]{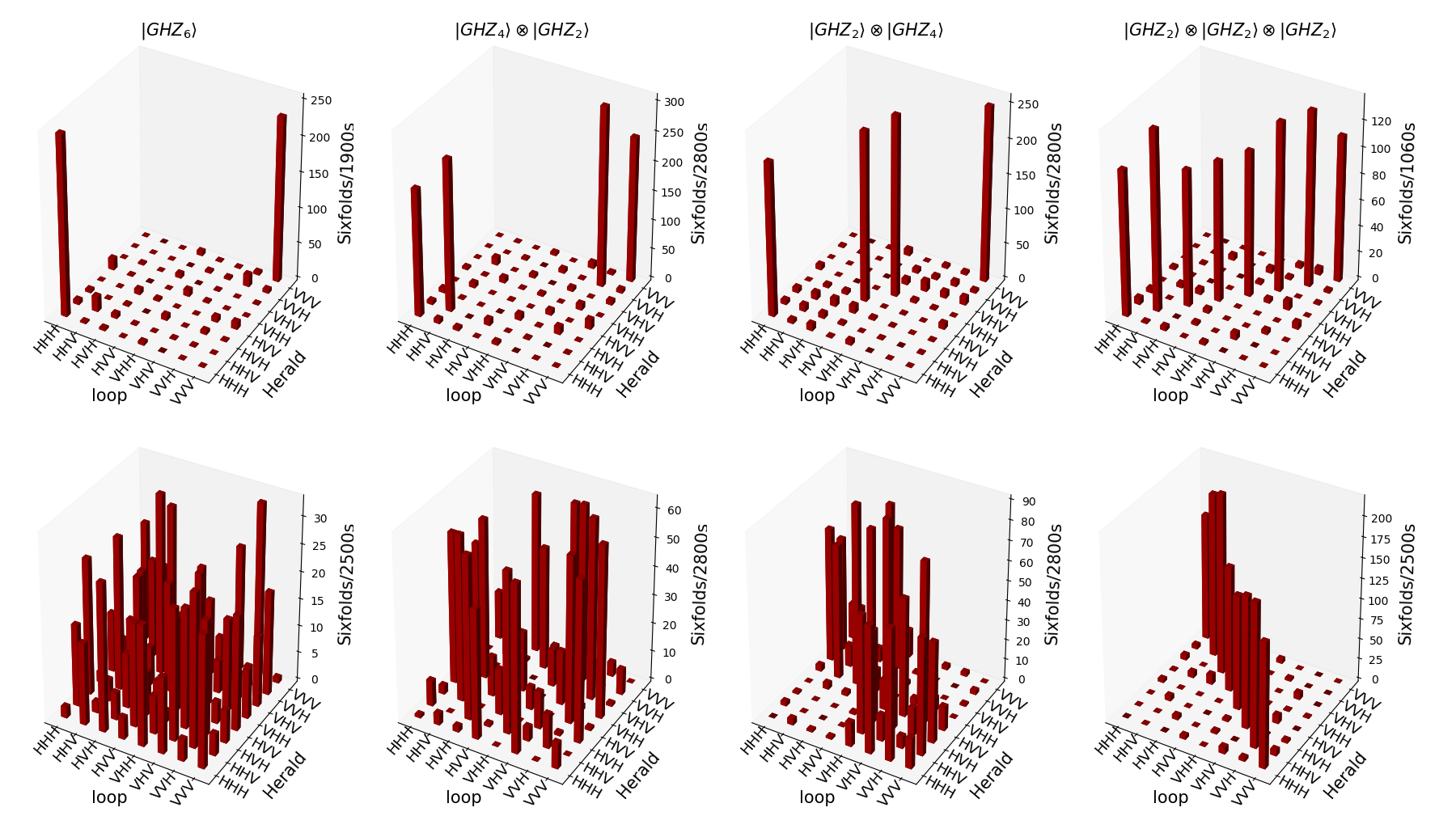}}
    \caption{%
        (Top row) Population of six-fold clicks for the four different topologies (as named in figure) measured in the H/V basis. 
        Direct population analysis shows $\mathcal{D}_N$ values of $(74\pm1.7)\%$, $(83.7\pm1.1)\%$, $(82.3\pm1.1)\%$, and $(92.6\pm0.8)\%$, respectively, for the states from left to right.
        A Poissonian error on the counts is included in the data analysis.
        As expected, fidelities decrease for larger-size GHZ states.
        (Bottom row) Six-fold click probabilities measured in a specific coherence term ($\langle\mathcal{M}_3^{\otimes 6}\rangle$) setting.
    }\label{fig:sppdata}
\end{figure*}

    For the described experimental approach, we expect the hypothesis testing scheme to be feasible for up to $N=8$ or $N=10$ qubits. This results from the following considerations: For the hypothesis test it is crucial to have GHZ fidelities of above $0.5$. Since the initial fidelity of a Bell pair generated by the source has roughly a fidelity of $0.92$, assuming naively that the fidelity of a GHZ state with $N$ qubits is upper bounded by $0.92^{(N/2)}$ leads to fidelities below $0.5$ already for a $18$-qubit GHZ state. However, the fidelity will be mostly far lower than this due to imperfections in the setup and the measurement settings, so the expected upper limit on the number of qubits will be around $N=8$ or $N=10$. Additionally, the generation and detection rates of larger GHZ states decrease rapidly. As described in the Supplemental Material of Ref.~\cite{Meyer2022}, in Section~S4.B, and assuming setup parameters similar to the state of the art publication \cite{zhong2018} we expect generation rates of $(6x10^4, 3x10^3, 3x10^2, 12, 0.6)$ Hz for $(4, 6, 8, 10, 12)$-qubit GHZ states. Since all of the photons must be detected, assuming a total loss of $50\%$ per photon leads to an exponential loss scaling. Combining the generation rate with the loss scaling, we end up with detection rates of $(3750, 46, 1.2, 10^{-2}, 10^{-4})$ Hz for the above states. Therefore, to collect $1000$ valid detection events, one must wait for around $(0.27, 22, 833, 10^{5}, 10^{7})$ seconds. Note that in the hypothesis testing scheme $N+1$ different measurement settings are needed, so if one wants $1000$ detection events per measurement setting this leads to another factor of $N+1$. All in all, this probably could be done for up to $N=10$ photons.

\section{Appendix E: Semi-Device-independent estimation on the fidelity}

    In our considerations before, we have assumed that the performed measurements are well calibrated and that the nodes are trusted. This, however, may not necessarily be the case. Therefore we consider here the device-independent scenario. In a device-independent scenario, there are no assumptions made on the measurements. 

    The key to discussing the device-independent scenario is to consider the Bell operator from the Mermin inequality \cite{Mermin1990}, which reads for three particles as
    \begin{equation}
        \mathcal{B}_3 = 
        X \otimes X \otimes X - 
        X \otimes Y \otimes Y -
        Y \otimes X \otimes Y -
        Y \otimes Y \otimes X. 
    \end{equation}
    Here, $X$ and $Y$ are general dichotomic observables, but not necessarily Pauli matrices.
    For local realistic models, $\mean{\mathcal{B}_3} \leq 2$ has to hold while the GHZ state reaches $\mean{\mathcal{B}_3} = 4$ if the measurements $X = \sigma_x$ and $Y = \sigma_y$ are performed.
    The Mermin inequality can be generalized to more particles;
    it consists of a combination of $X$ and $Y$ measurements with an even number of $Y$s and alternating signs.
    Formally, this can efficiently be written as $\mathcal{B}_N = [(X+iY)^{\otimes N}+(X-iY)^{\otimes N}]/2$, and, with the identification $X = \sigma_x, Y = \sigma_y$, one finds $\mathcal{B}_N = 2^{N-1} \mathcal{A}_N$ (see also Eq.~\eqref{eq:A_n}).

    The essential point is that several results connecting the expectation value $\mean{\mathcal{B}_N}$ with the GHZ-state fidelity are known.
    First, if nothing is assumed about the measurements, and if the value $\mean{\mathcal{B}_N}$ is close to the algebraic maximum $2^N$, a high GHZ fidelity up to some local rotations is certified \cite{Kaniewski2009}.
    Second, as we will show below, if one assumes that $X$ and $Y$ are (potentially misaligned) measurements on qubits, one can directly formulate a lower bound on the GHZ fidelity. This can be extended to other notions of misaligned measurements \cite{Cao2023}. 
    Third, it has been shown that, even in the presence of collaborating dishonest nodes, the global state must be close to a GHZ state if the honest nodes choose $X = \sigma_x$ and $Y = \sigma_y$ as measurements \cite{Pappa2012}.
    Finally, if all parties choose  $X = \sigma_x$ and $Y = \sigma_y$, then $F_N \geq \mean{\mathcal{B}_N} /2^{N-1}$ since for quantum states $\mean{\mathcal{D}_N} \geq \mean{\mathcal{A}_N}$.

    The above comments suggest the following scheme for the device independent scenario.
    All parties measure all combinations of $X$ and $Y$, leading to $2^N$ measurements in total.
    Then, they can evaluate the Bell operator of the Mermin inequality on each subset and characterize the fidelity for all sources with three or more qubits. 
    For the case that the maximal size of the GHZ states is known to be $M<N$, not all $2^N$ measurements need to be performed.
    Rather, a smartly chosen subset of these measurements suffices as the scheme of overlapping tomography allows to evaluate all combinations of $X$ and $Y$ on $M$-particle subsets with an effort increasing only logarithmically in $M$ \cite{cotler2020quantum}. Finally, if all the multipartite sources are identified, the parties can
    use two-qubit Bell inequalities to identify the structure of the bipartite sources. In addition, they may also consider advanced measurement schemes based on continuous stabilizers \cite{Mccutcheon2016}, or the Svetlichny inequality \cite{Murta2023}, which may be more efficient in the presence of dishonest parties.
    
    As discussed above, there are already results connecting the violation of a
    Mermin inequality by some state with its GHZ-state fidelity \cite{Kaniewski2009}. 
    In the following, using the experimentally observed violation of the Mermin inequality $S_N = \expval{\mathcal{B}_N}$, we will give a lower bound of the GHZ-state fidelity under the assumption that one performs (misaligned) measurements $X$ and $Y$ on qubits.

    Firstly, recall that the Bell operator of the Mermin inequality for $N$ qubits is 
    given by
    \begin{align}
        \mathcal{B}_N = [(X+iY)^{\otimes N}+(X-iY)^{\otimes N}]/2.
    \end{align}
    We then specify the observables as
    \begin{align}
        X &= \sigma_x,\\
        Y &= \sin(\theta_k)\sigma_x + \cos(\theta_k)\sigma_y,
    \end{align}
    where $\theta_k$ is the angle denoting the misalignment from a perfect measurement 
    ($X = \sigma_x$ and $Y = \sigma_y$), which can be different for each qubit $k = 1,...,N$.
    Note that we can assume $X = \sigma_x$ since we are only interested in the GHZ-state fidelity modulus local unitaries. For the same reason, it is no further restriction to assume the
    second measurement to be in the $x$-$y$ plane of the Bloch sphere.

    The goal is now to estimate the lower bound on the GHZ-state fidelity, given a violation 
    $S_N$ of the $N$-qubit Mermin inequality.
    For two qubits, this has previously been done using the spectral decomposition of the Clauser, Horne, Shimony, Holt (CHSH) operator \cite{Bardyn2009}. We now extend this method to $N>2$ qubits by using the spectral decomposition of the Mermin operator $\mathcal{B}_N$ and the fact that, for any Bell operator with two observables per qubit, its eigenstates (if not degenerate) are given by the $N$-qubit GHZ states, such that its eigendecomposition is $\mathcal{B}_N = \sum_i \lambda_i \ket{GHZ_i}\bra{GHZ_i}$ \cite{Scarani2001}. 

    The idea is now the following:
    If the violation exceeds a certain value $S_N$, we want to be able to state that the fidelity for some GHZ state $\ket{GHZ_i}$ is at least $F$.
    Thus, we start by fixing a fidelity $F$ and maximize the violation of the Bell inequality.

    Consider the expectation value of the Bell operator for some state $\varrho$,
    \begin{align}
        S_N &= \Tr(\mathcal{B}_N\varrho) = \sum_i \lambda_i \bra{GHZ_i}\varrho \ket{GHZ_i}\\
        &= \sum_i \lambda_i F(\varrho , \ket{GHZ_i})
        \leq \lambda_1 F + \lambda_2 (1-F),
        \label{eq:maxvio}
    \end{align}
    where we fixed the largest fidelity $F \coloneqq \max_i \{F(\varrho , \ket{GHZ_i})\}$ and sorted the eigenvalues $\lambda_i$ in decreasing order.
    
    As mentioned before, the measurements might be misaligned with respect to some angles $\theta_k$ and therefore the eigenvalues $\lambda_i$ are functions of the angles $\theta_k$.
    Since we are interested in the largest violation one can achieve for a fixed fidelity $F$, we have to consider all possible misalignments $\{\theta_k\}$.
    Then, Eq.~\eqref{eq:maxvio} reads
    \begin{align}
        S_N \leq \max_{\{\theta_k\}} [\lambda_1(\{\theta_k\}) F +\lambda_2(\{\theta_k\}) (1-F)].\label{eq:bound}
    \end{align}
    In order to find the maximum of this expression, consider
    \begin{align}
        \mathcal{B}_N^2 = [((X+iY)^2)^{\otimes N} + ((X-iY)^2)^{\otimes N}+((X+iY)(X-iY))^{\otimes N}+((X-iY)(X+iY))^{\otimes N}]/4.
    \end{align}
    Using $\Tr(\sigma_i\sigma_j) = 2\delta_{ij}$, it follows that $\Tr(XY) = 2\sin(\theta_k)$, and further that
    \begin{align}
        \Tr(\mathcal{B}_N^2) &= [\prod_{k=1}^N 4i\sin(\theta_k) + \prod_{k=1}^N (-4i)\sin(\theta_k) + 4^N +4^N]/4 \\
        &= [(4i)^N\prod_{k=1}^N \sin(\theta_k) + (-4i)^N\prod_{k=1}^N \sin(\theta_k) + 2\cdot4^N ]/4. 
    \end{align}
    For an odd number of qubits $N$, this reduces to 
    \begin{align}
        \Tr(\mathcal{B}_N^2) &= [2\cdot4^N ]/4 = 2^{-1}2^{2N} = 2^N2^{N-1}.
    \end{align}
    Following Ref.~\cite{Scarani2001}, where similar formulas were derived, we note that
    \begin{align}
        \sum_{i=1}^{2^N}\lambda_i^2 = \Tr(\mathcal{B}_N^2) = 2^N2^{N-1}.
    \end{align}
    Further, it is known that the eigenvalues of the Mermin operator $\mathcal{B}_N$ appear pairwise with alternating signs; therefore, summing up only the squared positive eigenvalues yields
    \begin{align}
        \sum_{i=1}^{2^{N-1}} (\lambda_i^+)^2 = 2^N2^{N-1}/2 = 2^{2(N-1)}.
    \end{align}
    Thus, the two largest eigenvalues must fulfill
    \begin{align}
        \lambda_1^2+\lambda_2^2 \leq 2^{2(N-1)}.
        \label{eq:lambdaCond}
    \end{align}
    We can use this to parameterize $\lambda_1 = 2^{N-1}\cos(\alpha)$ and $\lambda_2 = 2^{N-1}\sin(\alpha)$, as done in Ref.~\cite{Bardyn2009}.
    Note that $\alpha$ is not directly connected to the misalignment $\theta_k$ anymore, but it simply represents a parameter used to express all possible configurations for $\lambda_1$ and $\lambda_2$ saturating Eq.~\eqref{eq:lambdaCond}. 
    Then, Eq.~\eqref{eq:bound} yields
    \begin{align}
         S_N &\leq \max_{\alpha} [2^{N-1}\cos(\alpha) F + 2^{N-1}\sin(\alpha)(1-F)]\\
         &\leq 2^{N-1} \sqrt{F^2+(1-F)^2},
    \end{align}
    using that $\max_\alpha (a\sin(\alpha)+b\cos(\alpha)) = \sqrt{a^2+b^2}$, and it follows that
    \begin{align}
        F \geq \frac{1}{2} + \frac{1}{\sqrt{2}}\sqrt{\left(\frac{S_N}{2^{N-1}}\right)^2-\frac{1}{2}}
        \label{eq:anaF}
    \end{align}
    for an odd number of qubits.
    Note that this expression is only defined if the violation is large enough, which shows that, in the semi-device-independent scenario, a violation of at least $S_N = 2^{N-1}/\sqrt{2}$ is needed to certify entanglement.
    Furthermore, if the maximal violation $2^{N-1}$ is achieved, the GHZ-state fidelity is $F=1$.

\begin{figure}[tbp]
    \centering
    \includegraphics[width = 0.8\textwidth]{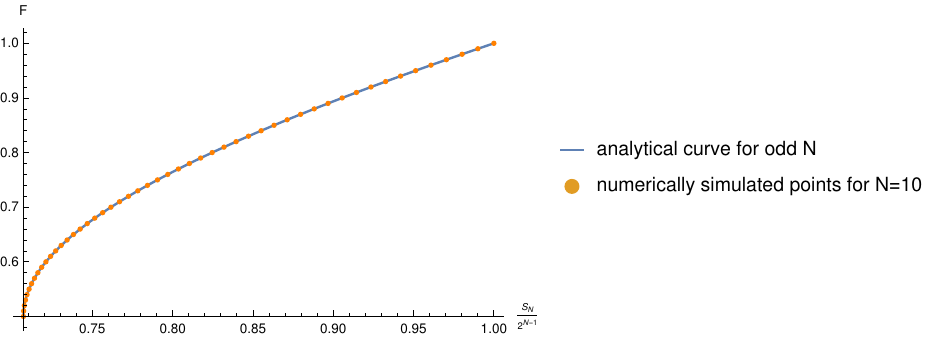}
    \caption{
        The minimum GHZ-state fidelity for a given violation of the $N$-qubit Mermin inequality. 
        We compare the analytically computed curve following Eq.~\eqref{eq:anaF} to the numerically simulated values for 10 qubits.
    }
    \label{fig:violationvsfid}
\end{figure}

    Lastly, it is to mention that this ansatz to find the analytical expression only works for an odd number of qubits.
    However, when numerically maximizing Eq.~\eqref{eq:bound} for up to 10 qubits and scaling the violation with a factor $2^{-(N-1)}$, we obtain the same curve for all $N$, as shown in Fig.~\ref{fig:violationvsfid}.
    Thus, we conjecture that Eq.~\eqref{eq:anaF} is true for all $N$.

    For the experimental implementation of the 
    semi-device-independent scenario one should note one 
    important difference to the device-dependent scheme 
    described in the main text: While in the device-dependent 
    scenario all parties measure the same $N+1$ local 
    measurement settings at the same time, here they would need to 
    measure all $2^N$ possible combinations of the dichotomic 
    observables $X$ and $Y$. For the experimental set-up described 
    in Appendix~D, this makes a significant difference. As can be 
    seen in Figure~\ref{fig:setup}, measuring the same measurement 
    setting on three of the qubits is straightforward as the same 
    detector can be used for the three photons used to trigger the 
    feed-forward signal, as well as for for the three photons stored 
    in the quantum memory. One possibility to measure different observables on these three qubits is to split the beam to two measurement setups (effectively: a POVM, which implements the $X$ and $Y$  measurement at the same time), and using the time information from the detectors to identify which measurement has
    been performed. For the available set-up, however, this would require significant effort in the data processing and additional detectors.

\end{document}